\documentclass[10pt,preprintnumbers,superscriptaddress ,amsmath,amssymb,twocolumn,pra]{revtex4-2}

\usepackage{graphicx}%
\usepackage{dcolumn}%
\usepackage{bm}%
\usepackage{amsmath}

\usepackage{soul}

\usepackage[colorlinks = true,
            linkcolor = blue,
            urlcolor  = blue,
            citecolor = blue,
            anchorcolor = blue]{hyperref}
\usepackage{pgfplots} 
\usepackage[font=small,labelfont=bf, justification=justified, format=plain]{caption}

\usepackage[utf8]{inputenc}
\usepackage{amsmath}
\usepackage{mathrsfs}
\usepackage{graphics}
\usepackage{xspace}
\usepackage[T1]{fontenc}
\usepackage{graphicx}%
\usepackage{dcolumn}%
\usepackage{bm}%
\usepackage[capitalise]{cleveref}

\usepackage{amsmath}

\usepackage{amsthm}
\theoremstyle{thmstyleone}%
\newtheorem{theorem}{Theorem}%

\theoremstyle{thmstyletwo}%

\theoremstyle{thmstylethree}%
\newtheorem{definition}{Definition}%
\newtheorem{assumption}{Assumption}
\newtheorem{lemma}{Lemma}

\usepackage{physics}

\usepackage[font=small,labelfont=bf, justification=justified, format=plain]{subcaption}
\usepackage{soul}

\usepackage{tikz}
\usepackage{graphicx}
\usetikzlibrary{shapes.geometric}
\tikzset{
  half circle/.style={
      semicircle,
      shape border rotate=180,
      anchor=chord center,
      minimum size=5mm
      }
}
\usetikzlibrary{arrows.meta} 
\usetikzlibrary{arrows}

\usepackage{booktabs}

\usepackage{graphicx}%
\usepackage{multirow}%
\usepackage{amsmath,amssymb,amsfonts}%
\usepackage{amsthm}%
\usepackage{mathrsfs}%
\usepackage[title]{appendix}%
\usepackage{xcolor}%
\usepackage{textcomp}%
\usepackage{manyfoot}%
\usepackage{booktabs}%
\usepackage[linesnumbered,ruled,vlined]{algorithm2e}

\usepackage{listings}%

\usepackage{natbib}

\newcommand{\indep}{\mathrel{\perp\!\!\!\perp}}

\renewcommand{\P}{\text{P}}

\newcommand{\Polylog}{\text{polylog}}

\DeclareMathOperator*{\argmin}{\arg\!\min}

\usepackage{graphicx}%
\usepackage{dcolumn}%
\usepackage{bm}%

\usepackage{ragged2e}

\makeatletter
\let\@reinserts\relax
\makeatother
\pgfplotsset{compat=1.18}

\begin{document}

\captionsetup[figure]{labelfont={bf},labelformat={default},labelsep=period,name={Fig.}}

\preprint{APS/123-QED}

\title{Quantum Algorithms for Causal Estimands}%

\author{Rishi Goel}
\email{uqrgoel2@uq.edu.au}
\affiliation{School of Mathematics and Physics, The University of Queensland, St. Lucia, QLD 4072, Australia}
\affiliation{ARC Centre for Engineered Quantum Systems, The University of Queensland, St. Lucia, QLD 4072, Australia}

\author{Casey R. Myers}
\affiliation{School of Physics, Mathematics and Computing, The University of Western Australia, Crawley, WA, 6009, Australia}
\affiliation{Pawsey Supercomputing Centre, 1 Bryce Avenue, Kensington, WA, 6151, Australia}

\author{Sally Shrapnel}
\affiliation{School of Mathematics and Physics, The University of Queensland, St. Lucia, QLD 4072, Australia}
\affiliation{ARC Centre for Engineered Quantum Systems, The University of Queensland, St. Lucia, QLD 4072, Australia}

\date{\today}%

\begin{abstract}
Modern machine learning methods typically fail to adequately capture causal information. Consequently, such models do not handle data distributional shifts, are vulnerable to adversarial examples, and often learn spurious correlations \cite{schölkopf2022statisticalcausallearning}. Causal machine learning, or causal inference, aims to solve these issues by estimating the expected outcome of counterfactual events, using observational and/or interventional data, where causal relationships are typically depicted as directed acyclic graphs. It is an open question as to whether these causal algorithms provide opportunities for quantum enhancement. In this paper we consider a recently developed family of non-parametric, continuous causal estimators and derive quantum algorithms for these tasks. Kernel evaluation and large matrix inversion are critical sub-routines of these classical algorithms, which makes them particularly amenable to a quantum treatment. Unlike other quantum machine learning algorithms, closed form solutions for the estimators exist, negating the need for gradient evaluation and variational learning. We describe several new hybrid quantum-classical algorithms and show that uniform consistency of the estimators is retained. Furthermore, if one is satisfied with a quantum state output that is proportional to the true causal estimand, then these algorithms inherit the exponential complexity advantages given by quantum linear system solvers.    
\end{abstract}

\maketitle

\section{Introduction}\label{Section: Introduction}

Causal learning has emerged as a recent and important extension to the field of classical machine learning~\cite{JanzingBook2017,schölkopf2022statisticalcausallearning}, having proven itself invaluable in many fields, including healthcare~\cite{SanchezRSOS2022}, social sciences~\cite{PanArXiv2024} and economics~\cite{BaiardiEconometricsJ2024}. Causal learning is effective because it enables models to move beyond mere correlations and uncover underlying cause-and-effect relationships. When applied in new settings that involve distributional shifts or interventions, traditional machine learning models often struggle since they rely on statistical associations that may not hold~\cite{ScholkopfIEEE2021}. While deep learning and generative machine learning have taken centre stage in the industrial application of automated learning~\cite{JamwalIJIMDI2022}, it is well known that these technologies have significant vulnerabilities due to their inability to reliably capture causal concepts~\citep{schölkopf2022statisticalcausallearning}. Increasingly, classical machine learning (ML) experts are taking techniques from causal learning, traditionally limited to small data sets of low dimensionality, and injecting modern ML elements to create new algorithms that can take a range of complex data types (e.g. graphs, images, text)~\cite{grettonkernelcausal} as input.

Quantum machine learning (QML) has taken off as a field since its early adoption a decade ago~\cite{RebentrostPRL2014,LloydNat2014,BiamonteNat2017}. It has demonstrated broad applicability across a range of areas, with promising results in data analysis in the form of clustering~\cite{KerenidisNEURIPS2019}, classification~\cite{AbohashimaArXiv2020} and dimensionality reduction~\cite{YuPhysicaA2023}, as well as for optimisation problems~\cite{BlekosPhysRep2024}, generative models~\cite{ZoufalArXiv2021}, quantum chemistry~\cite{TillyPhysRep2022} and finance~\cite{PistoiaIEEE2021}. Recently there has been a considerable effort to develop noisy intermediate-scale quantum (NISQ) algorithms for QML tasks in the form of variational circuits~\cite{CerezoNatRevPhys2021, BhartiRMP2022}, but there has also been significant progress in designing fault tolerant QML algorithms~\cite{GuoArXiv2024} that take advantage of quantum computing's intrinsic ability to speed-up linear algebra operations~\cite{MoralesArXiv2025}. Given the large number of proposed quantum machine learning algorithms \cite{CerezoNatCompSci2022}, it is surprising that ideas from causal learning have not been considered. Combining the approaches of QML and causal learning  %
provides an opportunity to find new avenues to realise potential quantum advantage that will be relevant to a broad range of fields.

In this paper we propose a new quantum algorithm for causal inference and show it can be used as a core subroutine to perform a variety of causal inference tasks. We examine convergence rates and also prove that uniform consistency of each of these estimators is retained. In the main body of the paper we present these results for estimating \emph{dose response}, however in the appendix we show the results also hold for (i) causal estimands that require causal sufficiency (no hidden common causes): \emph{dose response with distributional shift, conditional response} and \emph{heterogeneous response}, and (ii) causal estimands that do not require causal sufficiency: \emph{counterfactual mean outcome given intervention, counterfactual distribution given intervention} and \emph{counterfactual distribution embedding given intervention}. It is likely there will be many more possibilities and we hope this paper will provide impetus for fruitful new algorithmic discovery in this direction. As is often the case, our algorithms use quantum linear subroutines, so we calculate query complexity scaling relative to condition number and with the caveats of matrix sparsity and the possibility of fault tolerant hardware. Interestingly, the causal paradigm provides a framework in which one can tune a single regularisation hyperparameter to lower condition number, thus these algorithms are particularly well-suited for a quantum treatment.

{The paper is structured as follows: Section \ref{Section: Classical Causal Effect Estimation} describes the classical task of \emph{Causal Effect Estimation} and presents the corresponding convergence guarantees. Section \ref{Section: Quantum Algorithms for Causal Estimands} highlights computational bottlenecks for the classical algorithm and presents pseudocode for our proposed alternative quantum algorithm \emph{Quantum Dose Response}. This section also contains our main result which proves uniform consistency for this algorithm. Our consistency analysis is split into 2 sections. \ref{Section: Uniform Consistency and Convergence Rates of Quantum Algorithms} proves the uniform consistency guarantee for the quantum algorithm while \ref{Section: Convergence Rate with Quantum Measurement Error} further derives results for the case where one has measurement error. Finally, we conclude and make suggestions for future work in \ref{Section: Discussion}.}

\section{Classical Causal Effect Estimation}\label{Section: Classical Causal Effect Estimation}

Unlike traditional ML approaches which are built by identifying symmetric correlations, causal learning starts from an underlying assumption of directedness. Causal ML models aim to predict what will happen when one acts in the world, and therefore need to be robust to the kind of distributional shifts caused by interventions. Such models are particularly popular in medical and economic contexts where causal predictions frequently inform policy \cite{schölkopf2022statisticalcausallearning}. There is a large literature of distinct approaches, including the potential outcomes framework \cite{Richardson2013SingleWI}, the do-calculus and structural equations approach \cite{JudeaPearl1995, Sprites/Glymour}, among others (see figures in appendix \ref{Section: Causal Structure} for illustrative examples).
In this paper we take as our starting point the potential outcomes framework, where one can ask counterfactual questions, such as what \emph{would} have happened to the outcome $Y^{(a)}$ if we had tried treatment $a'$ instead of treatment $a$?

The goal of our causal learning algorithm is to estimate a causal function that captures the expected counterfactual outcome $Y^{(a)}$ given a hypothetical intervention that sets $A=a$, in the presence of confounder variables (see figure \ref{fig:3_node_structure_main_text}). The key causal function of interest is hence:

\begin{align}
    \text{Dose response:  }\ \theta_0(a)=\mathbb{E}\{Y^{(a)}\}.
\end{align}

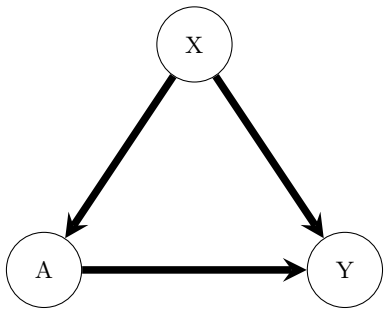
\begin{figure}[!htbp]
    \centering
    \begin{tikzpicture}
    \node[circle, draw, minimum size=1cm] (A) at (0, 0) {A};
    \node[circle, draw, minimum size=1cm] (X) at (2, 3) {X};
    \node[circle, draw, minimum size=1cm] (Y) at (4, 0) {Y};

    \draw[->, line width=1mm, >=stealth, scale=2] (X) -- (A);
    \draw[->, line width=1mm, >=stealth, scale=2] (X) -- (Y);
    \draw[->, line width=1mm, >=stealth, scale=2] (A) -- (Y);
\end{tikzpicture}
        \caption{\justifying Three node causal Directed Acyclic Graph (DAG). Here $X$ represents confounding covariates: any variable that has a causal effect on both the treatment $A$ and the outcome $Y$. There is an assumed causal relationship between $A$ and $Y$.}
        \label{fig:3_node_structure_main_text}
\end{figure}

While parametric estimators for low dimensional discrete variables have been available for these causal functions since 1978 \cite{rubin1978}, new algorithms have been developed that make no parametric assumptions on the underlying data distributions and can handle both discrete and continuous variables.  The standard three assumptions underlying many causal methods apply in this case: \emph{no interference, conditional exchangability}, and \emph{overlap} (see appendix \ref{Section: Classical Causal Inference Proofs} for more detail).

These new non-parametric approaches use the machinery of Reproducing Kernel Hilbert Spaces (RKHS) and Kernel Ridge Regression to derive a closed form expression to estimate Dose Response \cite{grettonkernelcausal} (see appendix \ref{Section: Classical Causal Inference Proofs} for the derivation):

\begin{align}
\hat{\theta}(a) &= n^{-1} Y^T (K_{AA} \odot K_{XX} +  n\lambda I)^{-1} (K_{Aa} \odot \sum_{x_i} K_{Xx_i}),\label{alg}\end{align}
where $n$ is the number of training data points, $Y$ is a vector of labelled outcomes, $\odot$ is the Hadamard product, $K_{AA}$ is the kernel matrix with each element $K_{AA}(i,j) = k(a_i,a_j)$ and similarly for $K_{XX}$. $\lambda$ is a regularising hyperparameter and the vectors $K_{Aa}$ are the column vectors of the kernel elements between the set $A$ and the element $a$, and similarly for $K_{Xx_i}$.

Interestingly, it can be shown that the classical causal algorithm is uniformly consistent (valid across the full range of inputs) and converges for large $n$, with high probability. The rate of convergence{~\cite{grettonkernelcausal}} is dependent on the smoothness of the function, captured by a parameter $c$, and the effective dimension of the RKHS under consideration, captured by the parameter $b$:

\begin{align}
    \|\hat{\theta} - \theta_0\|_{\infty} &= \mathcal{O} \left[n^{-(c-1)/\{2(c+1/b)\}}\right]\label{unif consistent equation ATE main text},
\end{align}
where $c \in (1,2]$ and $b \in \mathbb{N}$ (higher numbers imply greater smoothness which requires less data to estimate). In the optimal case, $c=2, b\to \infty$, the convergence rate limits to $\mathcal{O}(n^{1/4})$. Further explanation is presented in appendix \ref{Section: Classical Uniform Consistency Results}.

\section{Quantum Algorithms for Causal Estimands}\label{Section: Quantum Algorithms for Causal Estimands}

There are two computational bottlenecks in implementing these causal estimation algorithms using classical hardware. The first is the evaluation of the $n^2$ entries in each kernel. While no quantum method exists to reduce the complexity of kernel evaluation, understanding the potential benefits of using quantum hardware to approximate kernel entries is a rapidly growing area of research---we consider these opportunities in appendix \ref{Section: Quantum kernel evaluation}. The second is the computational cost of multiplying and inverting large kernel matrices, a problem for which quantum algorithms are well suited. The classical complexity of exactly calculating a matrix inverse scales as $\mathcal{O}(n^3)$, where $n$ is the number of samples in the training set---LU decomposition or Gaussian Elimination methods are examples of such algorithms. {Alternatively, the conjugate gradient method scales linearly in n, $\sqrt{\kappa}$ and log(1/$\epsilon$), but is an approximate and iterative method.}

While it is well known~\cite{HarrowPRL2009, MoralesArXiv2025} that quantum linear subroutines can provide exponential speedups over classical counterparts when one assumes efficient quantum encoding and decoding strategies, it is not yet known whether statistical guarantees of uniform consistency and fast convergence are maintained for the particular application we present here. We next propose an algorithm that uses a quantum linear system subroutine for matrix inversion for calculating the causal estimand and analyse its statistical properties.

The goal of quantum linear system (QLS) solvers is to produce a quantum state $\ket{x}$ proportional to the solution of linear system $Ax = b$, where $A$ is an $n$-by-$n$ matrix. All have $\mathcal{O}(\Polylog (n))$ scaling in sample size, and additional scaling properties that depend on the condition number $\kappa$ of the matrix $A$~\cite{MoralesArXiv2025}. All these solvers assume efficient state preparation of $\ket{b}$ \st{and} {while only some require} block encoding of the operator $A$. Whether or not this will require QRAM is still an open question and will likely depend on the problem. However, recent work~\cite{GreenArXiv2025, CampsArXiv2023, MelnikovQST2023} on loading structured data into quantum circuits using matrix product states is very encouraging for the types of data set relevant to causal learning.

The following is pseudo-code for the implementation of a hybrid quantum algorithm for causal effect estimation of the estimand specified in Eq.~(\ref{alg}).

\begin{algorithm}[!htbp]
\caption{Quantum Dose Response $\hat{\theta}_0(a)$}\label{alg 1}
\KwIn{Kernel entries: $K_{AA}, K_{XX}, K_{Aa}, \sum_{x_i} K_{Xx_i}$; regularization parameter $\lambda$; outcome vector $Y$.}
\KwOut{Approximate quantum state with amplitude proportional to $\hat{\theta}_0(a)$}

\textbf{Classical Computation:} \\
Compute $A = K_{AA} \odot K_{XX} + n\lambda I$\;
Compute and normalize $b = \frac{1}{n} K_{Aa} \odot \sum_{x_i} K_{Xx_i}$\;

\textbf{Quantum Subroutine:} \\
Prepare quantum states $\ket{b}$ and $\ket{Y}$ and block encoding $U_A$\;
Use a quantum linear solver to obtain $\ket{A^{-1} b}$ to additive error $\epsilon$\;
Measure overlap $ \bra{Y} \ket{A^{-1} b}$\;
\end{algorithm}

It has recently been shown that if matrices are low rank, well-conditioned and high dimensional, quantum inversion methods only offer a polynomial advantage over classical algorithms \cite{Arrazola_2020}. Hence the opportunity for quantum algorithms is for high rank, high dimensional problems~\cite{HarrowPRL2009, MoralesArXiv2025}---fortunately this is commonly encountered in causal ML applications.

In addition, the computational complexity of an optimal quantum linear solver scales linearly with the condition number $\kappa$~{~\cite{HarrowPRL2009, MoralesArXiv2025}}. Consequently, the condition number of the matrix $A= K_{AA} \odot K_{XX} + n\lambda I$ (used in Alg. \ref{alg 1}) plays a critical role in determining the potential for exponential speedup. Notably, the regularisation parameter $\lambda$ and number of training samples $n$ directly influence the conditioning of this matrix.

As $K_{AA} \odot K_{XX}$ is diagonalizable, the addition of the regularisation term uniformly shifts its eigenvalues. Specifically, for any eigenpair $(\textbf{v}, e)$ of $K_{AA} \odot K_{XX}$, we have,

\begin{align}
(K_{AA}\odot K_{XX}) \textbf{v} &= e \textbf{v}, \\
(K_{AA}\odot K_{XX} + n\lambda I) \textbf{v} &= e \textbf{v} + n\lambda \textbf{v} = (e+n\lambda) \textbf{v}\label{kappa lambda equation dependence}.
\end{align}

Hence, the condition number of the regularised matrix is given by,

\begin{equation}
\kappa = \frac{e_{\max} + n\lambda}{e_{\min} + n\lambda},
\end{equation}
which approaches $1$ as $n\lambda$ increases. This demonstrates that for fixed $n$, regularisation can be used to improve conditioning, thereby reducing the query complexity of quantum linear solvers. Alternatively, if one fixes $\lambda$, adding more training data can also improve conditioning.
{Ultimately, the minimal $\lambda$ and $n$ for well conditioned matrices will depend on $e_{\max}$ and $e_{\min}$ which will be problem dependent.}

To empirically validate this behaviour, we investigate two causal datasets (Job Corp and Colangelo \cite{grettonkernelcausal}) with a fixed number of training samples. Using a Gaussian kernel, we calculate the kernel matrices, vary $\lambda$ between $0.005$ and $1$, and numerically calculate the respective condition number of $K_{AA}\odot K_{XX} + n\lambda I$. For this dataset we observe a strong sensitivity of the condition number to the regularisation parameter around a small neighbourhood of $\lambda$  (see Fig. \ref{fig:ate combined plot kappa lambda scaling}). Given one can tune $\lambda$ directly, this provides a practical means to tune $\kappa$ to achieve favourable computational scaling.

\begin{figure}
    \centering
    \includegraphics[width=\linewidth]{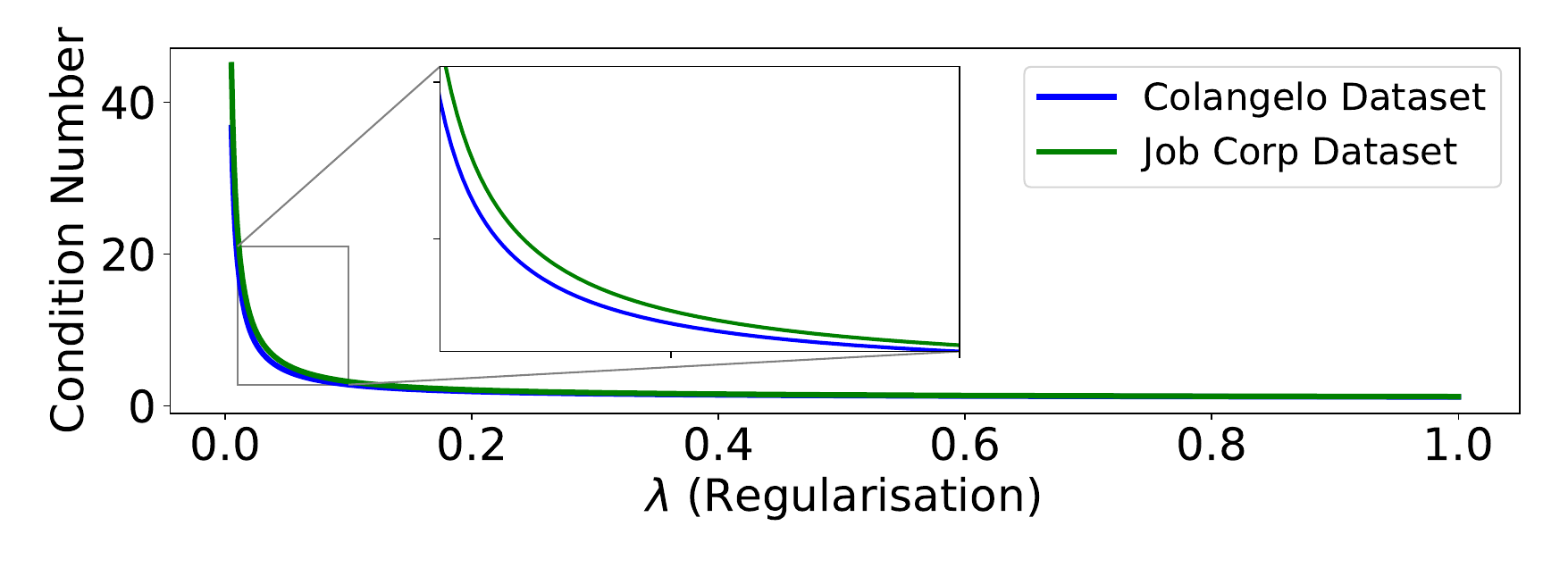}
    \caption{\justifying Empirical demonstration of the sensitive dependence between condition number $\kappa$ and regularisation $\lambda$ on Job Corp and Colangelo datasets for a small neighbourhood around $\lambda$. This indicates that one can use regularisation to tune the complexity scaling of the matrix inversion.}
    \label{fig:ate combined plot kappa lambda scaling}
\end{figure}

\subsection{Uniform Consistency and Convergence Rates of Quantum Algorithms}\label{Section: Uniform Consistency and Convergence Rates of Quantum Algorithms}

The quantum algorithm \ref{alg 1} introduces an additive error ($\epsilon$) due to the fact that the matrix inversion is approximate. Here we will consider whether this error leads to loss of convergence and uniform consistency guarantees. We show that one can simultaneously maintain tight uniform consistency bounds and retain the exponential speedup from matrix inversion. We extend our results to other causal estimands in appendix \ref{Section: Estimation Convergence with Quantum Algorithms}.

Application of an QLS algorithm yields a normalised quantum state $\ket{x} = \ket{A^{-1}b}$ to some additive error $\epsilon$. Formally, we write this as,

\begin{equation}
    \norm{{\ket{A^{-1}b}}_{quantum} - \ket{A^{-1}b}} \leq \norm{A^{-1}b} \epsilon\label{eqn:EpsilonError},
\end{equation}
where ${\ket{A^{-1}b}}_{quantum}$ represents the normalised approximate state found by the quantum inverse and ${\ket{A^{-1}b}}$ is {the ideal quantum representation of} the normalised exact classical vector. As expected, increasing the number of oracle calls reduces the additive error $\epsilon$. 
To achieve an error $\epsilon$, a minimum of $\mathcal{O}(\log(1/\epsilon))$ oracles calls is necessary {and arises from the filtering/projection stage}, as shown by the discrete-adiabatic theorem \cite{PRXQuantum.3.040303} and kernel refection~\cite{DalzellArXiv2024}.

Given exact classical linear system solvers have scaling $\mathcal{O}(n^3)$, an exponential speedup would require the quantum linear solver to run in $\mathcal{O}(\Polylog(n))$ {given a suitable downstream observable}~\cite{HarrowPRL2009}. However, the quoted runtime of quantum linear solvers depends on $\epsilon$ and $\kappa$. For an exponential advantage we thus require the scaling $\mathcal{O}(\log(1/\epsilon))$ be no greater than $\mathcal{O}(\Polylog(n))$, and therefore $\epsilon \in {\Omega}(e^{-\Polylog (n)})$, where $\Omega$ signifies a lower bound. %

Given this lower bound, we consider whether our quantum causal estimation algorithm still yields uniform consistency. We show that indeed in the asymptotic case, this additive error from a quantum inverse can be small enough to maintain uniform consistency, yet still give a speedup in the matrix inversion over classical methods.

\begin{theorem}
\label{consistency of ATE theorem}(Uniform Consistency For Quantum Causal Dose Response)\\
    Given algorithm \ref{alg 1}, there exists a function class for the additive error $\epsilon$ such that exponential speedup and uniform consistency guarantees are retained{, where a suitable downstream observable is specified}. Additionally, with this function class the asymptotic convergence rates remain the same as for the classical algorithm.
\end{theorem}

\begin{proof}

    We define $\hat{\theta}_{quantum}(a)$ as the quantum approximation of the classical estimand $\hat{\theta}(a)$, and we define $\theta_0(a)$ as the true dose response. We begin by noting that the absolute difference between the quantum estimation and the true value can be split by application of the Cauchy-Schwarz inequality. So with high probability we have,

\begin{align}
    &|\hat{\theta}_{quantum}(a) - \theta_0(a)|\\
    &= |\hat{\theta}_{quantum}(a) - \hat{\theta}(a) + \hat{\theta}(a) - \theta_0(a)|   \\
    &\leq|\hat{\theta}_{quantum}(a) - \hat{\theta}(a)| + |\hat{\theta}(a) - \theta_0(a)| \\
     &{\stackrel{Eq.\, (\ref{unif consistent equation ATE main text})}{\leq}}
    |\langle Y, A^{-1}b\rangle_{quantum} - \langle Y, A^{-1}b\rangle| + \mathcal{O}(n^{-\frac{1}{2}\frac{c-1}{c+1/b}}) \,\, \\
    &\leq 
    \norm{Y} \norm{\ket{A^{-1}b}_{quantum} -\ket{A^{-1}b}} + \mathcal{O}(n^{-\frac{1}{2}\frac{c-1}{c+1/b}})\\
    &{\stackrel{Eq.\, (\ref{eqn:EpsilonError})}{\leq}}
    \norm{Y} \norm{A^{-1}b} \epsilon + \mathcal{O}(n^{-\frac{1}{2}\frac{c-1}{c+1/b}})\label{ate consistency final eq no sampling error 1}
\end{align}
where we use the fact that $\norm{Y}$ and $\norm{A^{-1}b}$ are bounded and do not scale with $n$. As a result, equation \ref{ate consistency final eq no sampling error 1} is true for the supremum norm over $a$ as $\norm{Y},\norm{A^{-1}b}$ and $\epsilon$ are all independent of $a$. Hence, in the asymptotic limit, the worst case accuracy scaling of the quantum algorithm is determined by the greater of $\mathcal{O}(\epsilon)$ and $\mathcal{O}(n^{-\frac{1}{2}\frac{c-1}{c+1/b}})$. Firstly we note that $\epsilon$ converges to $0$ almost surely, hence retaining uniform consistency of the quantum algorithm. Secondly, one can always find an $n$ such that

\begin{equation}
  K_1 e^{-\Polylog (n)} < K_2 n^{-\frac{1}{2}\frac{c-1}{c+1/b}},  
\end{equation}
for any $K_1, K_2 \in \mathbb{R}$. This means that,

\begin{align}
    \Omega(e^{-\Polylog (n)}) \cap \mathcal{O}(n^{-\frac{1}{2}\frac{c-1}{c+1/b}}) \neq \phi\label{eq:NonZeroOverlap},
\end{align}
where $\phi$ represents the empty set. Therefore there exists a class of functions for which $\epsilon$ can both be in $\Omega(e^{-\Polylog (n)})$ (yielding exponential speedup {following selection of a suitable downstream observable}) and in $\mathcal{O}(n^{-\frac{1}{2}\frac{c-1}{c+1/b}})$ (yielding uniform consistency). This is illustrated for intuition in figure \ref{fig:converg2}.

\end{proof}

Note our proof uses the optimal classical scaling which has strong smoothness assumptions, if we relax this assumption, it may well be that the quantum algorithm has faster convergence rates.

\begin{figure}[!htbp]
    \centering
    \includegraphics[width=0.8\linewidth]{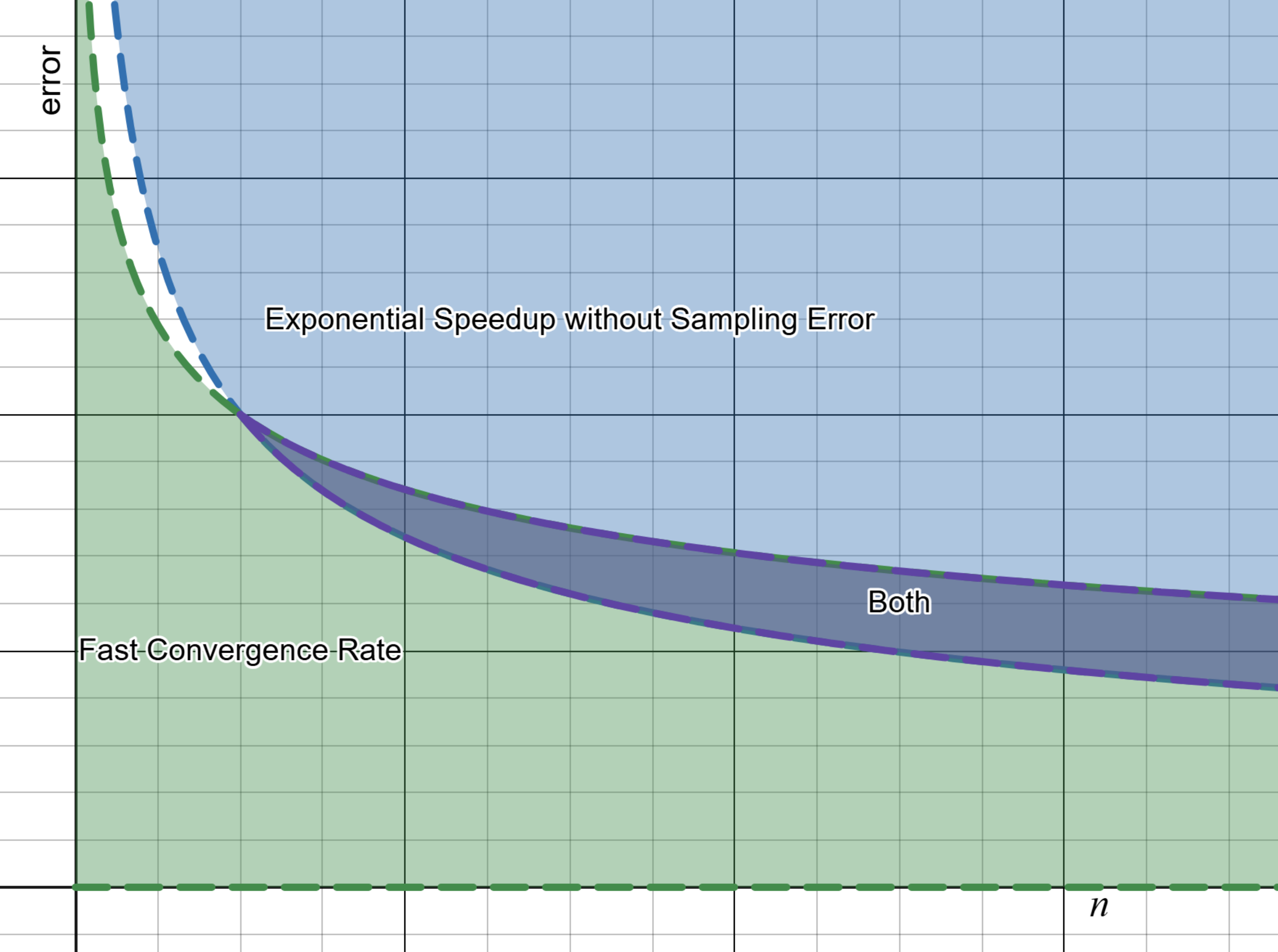}
    \caption{\justifying Illustration of the convergence rate function class against sample size.
    We see that within the blue region{---assuming a suitable downstream observable ---} one can achieve exponential speedup with $\log(1/\epsilon) = \mathcal{O}(\Polylog (n))$, which means that the error $\epsilon = \Omega(e^{-\Polylog (n)})$. This provides a lower bound on the error to maintain exponential speedup. In the green, one achieves fast convergence rate of $n^{-1/4}$ as determined via the best case classical uniform consistency convergence rate. In the purple region, we can achieve both a speedup and fast convergence as the additive error decreases fast enough. This is what is meant by the non empty subset where the functions classes overlap in Eq.~(\ref{eq:NonZeroOverlap}). Hence we can say that for large enough $n$, the algorithm maintains the same convergence rate and an exponential speedup.}
    \label{fig:converg2}
\end{figure}

\subsection{Convergence Rate with Quantum Measurement Error}\label{Section: Convergence Rate with Quantum Measurement Error}

Thus far we have considered the case where the quantum algorithm produces a quantum state proportional to the causal estimand. Such a situation is optimal when one wishes to utilise the quantum output directly as part of a downstream quantum enhanced workflow. If instead we wish to obtain a classical output, we must pay the cost of measurement. 

Mathematically, the estimation of an inner product via quantum measurements gives an error,

\begin{equation}\label{additive sampling error}
    \norm{\bra{X}\ket{Y} - \langle X, Y \rangle} \leq \epsilon_k,
\end{equation}
where $\bra{\cdot}\ket{\cdot}$ and $\langle \cdot, \cdot \rangle$ represent a quantum overlap estimation and an exact overlap, respectively. Current best methods (assuming no \emph{a priori} inference), such as the swap or Hadamard tests, require $1/\epsilon_k^2$ shots \cite{swap, McCleanNJP2016}. If we allow prior knowledge or feedback, we could use Bayesian inference or adaptive querying to achieve a $1/\epsilon_k$ scaling~\cite{HuszarPRA2012, FerriePRL2014}. Hence, at worst the query complexity of the quantum algorithm is multiplied by $1/\epsilon_k^2$. It is important to note that this limitation is not specific to the algorithms presented in this work, it applies to all quantum algorithms that rely on expectation value/kernel matrix evaluations. To retain the exponential speedup, we require the algorithm to run in $\mathcal{O}(\Polylog(n))$ and hence require $1/\epsilon_k^2 \in \mathcal{O}(\Polylog(n))$. Equivalently, this places a lower bound on the additive sampling error for exponential speedup of

\begin{equation}
     \epsilon_k \in \Omega (1/\Polylog (n)).
\end{equation}
Again, we consider how this affects uniform consistency and convergence rates (Theorem \ref{unif consistency with sampling error theorem}). 

\vspace{10pt}

\begin{theorem} \label{unif consistency with sampling error theorem}(Uniform Consistency of Dose Response with Sampling Error)\\
    Given algorithm \ref{alg 1}, there exists a function class for the additive error $\epsilon$ and sampling error $\epsilon_k$ such that the algorithms simultaneously retain exponential speedup and uniform consistency guarantees. However, there exists no function class for $\epsilon_k$ such that the asymptotic convergence rates remain the same as the classical algorithm.
\end{theorem}

\begin{proof}

Following from before, with the addition of equation \ref{additive sampling error}, we see that with high probability,

\begin{align}
    &|\hat{\theta}_{quantum}(a) - \theta_0(a)|\\ &= |\hat{\theta}_{quantum}(a) - \hat{\theta}(a) + \hat{\theta}(a) - \theta_0(a)|   \\
    &\leq|\hat{\theta}_{quantum}(a) - \hat{\theta}(a)| + |\hat{\theta}(a) - \theta_0(a)| \\
    &\leq 
    |\bra{Y}\ket{A^{-1}b}_{quantum} - \langle Y, A^{-1}b\rangle_{quantum} + 
    \langle Y, A^{-1}b\rangle_{quantum} \nonumber \\& - \langle Y, A^{-1}b\rangle| + \mathcal{O}(n^{-\frac{1}{2}\frac{c-1}{c+1/b}})\\
     &\leq 
    |\bra{Y}\ket{A^{-1}b}_{quantum} - \langle Y, A^{-1}b\rangle_{quantum}| \nonumber \\&+ 
    |\langle Y, A^{-1}b\rangle_{quantum} - \langle Y, A^{-1}b\rangle| + \mathcal{O}(n^{-\frac{1}{2}\frac{c-1}{c+1/b}})\\
    &\leq 
     \epsilon_k + 
    \norm{Y} \norm{\ket{A^{-1}b}_{quantum} -\ket{A^{-1}b}} + \mathcal{O}(n^{-\frac{1}{2}\frac{c-1}{c+1/b}})\\
    &\leq
     \epsilon_k +  
    \norm{Y} \norm{A^{-1}b} \epsilon + \mathcal{O}(n^{-\frac{1}{2}\frac{c-1}{c+1/b}}).\label{ate consistency final eq no sampling error}
\end{align}
Again this is true for the supremum norm and independent of $a$, meaning that the total asymptotic scaling is dependent on the greater of $\epsilon_k,\epsilon$ and  $\mathcal{O}(n^{-\frac{1}{2}\frac{c-1}{c+1/b}})$. All three terms converge to $0$ almost surely, implying uniform consistency. 

However in terms of convergence rates, we can always find an $n$ such that for any $c \in (1,2]$ and $b \in \mathbb{N}$,

\begin{equation}
    1/\Polylog (n) > \mathcal{O}(n^{-\frac{1}{2}\frac{c-1}{c+1/b}}).
\end{equation}
Hence, we can see that, 

\begin{equation}
    \Omega(1/\Polylog (n)) \cap \mathcal{O}(n^{-\frac{1}{2}\frac{c-1}{c+1/b}}) = \phi.
\end{equation}
Consequently, the sup norm convergence reduces to,

\begin{equation}
    |\hat{\theta}_{quantum}(a) - \theta_0(a)| \leq \epsilon_k \in \Omega(1/\Polylog (n)).
\end{equation}
In any practical application of this quantum algorithm one would take the smallest family for $\epsilon_k$, meaning that 

\begin{equation}
    |\hat{\theta}_{quantum}(a) - \theta_0(a)| \in \Theta(1/\Polylog (n)), \label{best case consistency}
\end{equation}
which does converge to $0$, albeit exponentially slower than the respective classical convergence rate. This means that our quantum algorithm retains uniform consistency, but has a slower convergence rate than the classical algorithm. However, in cases where one cannot assume a smooth function and a low effective RKHS dimension (e.g. c=1), how the classical and quantum convergence rates compare is an open question.  We illustrate this intuition in figure \ref{fig:converg1}. We also extend these results to other causal estimands in appendices \ref{Section: Estimation Convergence with Quantum Algorithms} and \ref{Incremental Functions, Counterfactual Distributions and Graphical Models}.

\end{proof}

\begin{figure}[!htbp]
    \centering
    \includegraphics[width=0.8\linewidth]{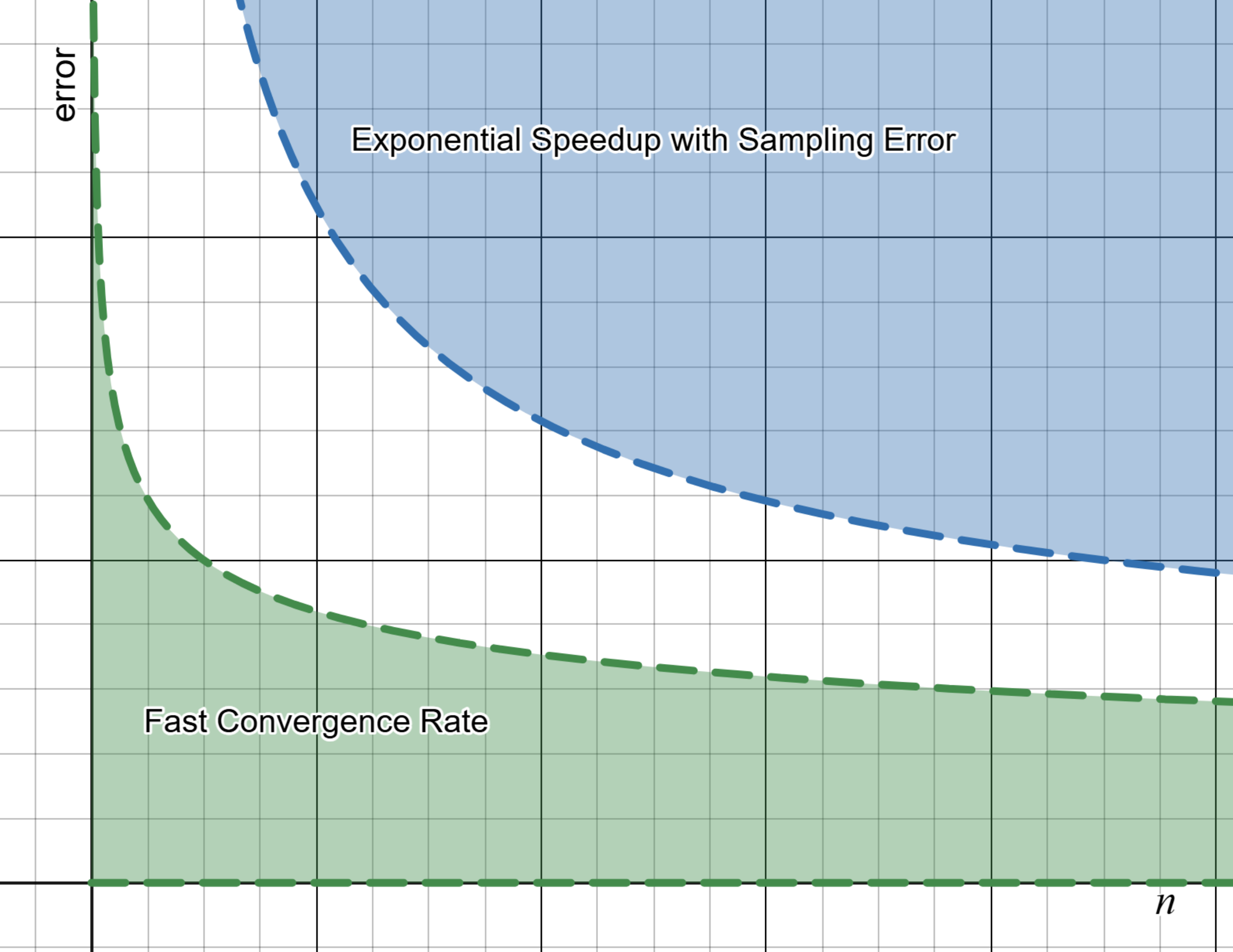}
    
    \caption{\justifying Illustration of the convergence rate function class against sample size that also considers sampling error. We see that the region for quantum exponential speedup (blue) is now disjoint from the classical fast convergence rate (green). This is because the minimum error we can permit with exponential speedup (blue) is always greater and decays slower than the maximum bound on classical fast convergence rate (green). Note this is for the best case classical algorithm where we can assume strong smoothness properties of the function to be learned and low effective RKHS dimension. This may not be true when this smoothness assumption is invalid.}
    \label{fig:converg1}
\end{figure}

\section{Discussion}\label{Section: Discussion}

In this work, we have taken the first steps towards integrating quantum computing techniques into the domain of causal inference, a key sub-field of machine learning that focuses on predicting the outcomes of interventions. Our primary contribution is the development of a quantum-assisted algorithm for a non-parametric, kernel-based causal estimator, which traditionally relies on complex matrix inversion and kernel evaluation. By adapting these methods to leverage quantum linear system solvers, we have shown that, under the assumption of efficient block encoding and state preparation, it is possible to retain the uniform consistency and convergence guarantees of classical algorithms while also achieving computational advantages.

A central finding is that the quantum variants of these estimators inherit the strong theoretical properties of their classical counterparts. Specifically, we demonstrated that if the underlying causal relationships and data-generating mechanisms are sufficiently smooth, the quantum algorithm can match the best-known classical convergence rates. Moreover, the quantum linear algebra subroutines we employed offer the potential for exponential speedups in computational complexity, compared to the polynomial scaling of classical matrix inversion techniques.

That said, our analysis also highlights some of the limitations and trade-offs in practical applications. For example, extracting classical values from quantum states inevitably introduces sampling (measurement) errors that can slow down the effective convergence rate. Thus, the full quantum advantage may be tempered by the resources needed to read out the results. It is possible that in the case when $c<2$ in Theorem~\ref{unif consistency with sampling error theorem}, there could be cases where the quantum algorithm has the same convergence rate as the classical algorithm, something that will be addressed in future work. Nonetheless, our results point to a promising avenue: if the downstream application can utilize quantum state outputs directly—e.g., as part of a larger quantum-enhanced workflow—then these measurement overheads can be minimised or avoided. For example, recent promising work~\cite{CatliArXiv2025, ChakrabartiArXiv2025} on encoding difficult optimisation problems directly within quantum systems, solving the desired problem coherently, without the need of a classical optimsiation step, as in QAOA~\cite{BlekosPhysRep2024}, suggest that it is possible to utilise quantum states as inputs for follow-on (quantum) processes in the quantum machine learning pipeline.

Our results transfer to other causal algorithms including dose response with distributional shift, average treatment on treated and conditional average treatment effect as well as other problem classes such as incremental functions, counterfactual distributions and graphical models (see Appendix \ref{Incremental Functions, Counterfactual Distributions and Graphical Models}). 

In addition, our framework involves a single hyperparameter that indirectly controls the conditioning of the kernel matrix, thereby directly determining the quantum algorithm’s scaling properties. This interplay between hyperparameter tuning, condition number, and algorithmic complexity is an interesting insight that we will explore in future work.

Overall, our work serves as a proof-of-concept that quantum computing can enhance the scalability of advanced causal inference methods without compromising their theoretical performance guarantees. It opens the door to future research on a range of open questions: What classes of causal structures are most amenable to quantum speedups? Will similar theoretical guarantees hold for more complex causal models? We hope this work will stimulate ongoing innovation at the intersection of causal inference and quantum machine learning.

\begin{acknowledgments}
This work has been supported by the Australian Research Council (ARC) Centre of
Excellence for Engineered Quantum Systems (EQUS, CE170100009). We would like to thank Riddhi Gupta and Dominic Berry for helpful discussions.
\end{acknowledgments}

\bibliography{citations}%

\onecolumngrid

\appendix

\newpage

\section*{Appendices}

\emph{Note that Appendices A, B and C are not original work and are included for pedagogical reasons to ensure the paper is self-contained. Appendices D, E and F contain original work and proofs to support the claims of the paper.}

\section{Relevant Background For Classical Causal Inference}\label{Section: Causal Structure}

The field of causality, particularly the study of necessary requirements to \textit{infer} causal relationships, has been mathematically grounded \cite{JudeaPearl1995, Sprites/Glymour, Richardson2013SingleWI, RubinPotentialOutcomes}. One can formally define a causal effect as follows.

\begin{definition}(Causal Effect)\label{Causal Effect}\\
    We define a random variable $X$ having a direct causal effect on another random variable $Y$ if two distinct interventions (i.e. setting $X=x,x'$) have different resulting distributions for $Y$. Colloquially, one may interpret this as a statement formalising that forcing changes in $X$ results in meaningful changes in $Y$. 
\end{definition}

A graphical approach to visualise causal relationships provides a compact way to encode such causal dependencies between random variables. These models are based on directed acyclic graphs (DAGs), where nodes represent random variables and edges represent causal dependence (Definition \ref{Causal Effect}). In DAGs, each node has a set of parents and descendents. We define parents of $X$ as the nodes which have a causal effect on $X$ and descendents as all nodes for which $X$ has a direct or indirect causal effect on.
Part of this definition of causality is the Markovian assumption, stating that the probability distribution of nodes are conditionally independent of their non-descendents given their parents. This is equivalent to factorising the joint probability distribution as

\begin{equation}
    \P(X_1,\dots,X_n) = \prod_{i=1}^n \P(X_i|\textbf{PA}_i),
\end{equation}
where $\textbf{PA}_i$ is the set of parents of $X_i$ or equivalently the nodes of direct causes of $X_i$. The factorization of joint probability distributions in DAGs naturally leads to the potential outcomes framework, particularly as developed by James Robins for causal inference in complex settings \cite{Richardson2013SingleWI}. Since DAGs encode causal relationships through conditional independence assumptions, they provide a structural foundation for reasoning about interventions. The potential outcomes framework formalizes causal effects by considering counterfactual scenarios—what would have happened had a different treatment or action been taken. This allows us to compare causal effects by asking how under precisely the same conditions, two different treatments would have differed. The fundamental problem of causal inference is that you can only ever observe one of the two treatments for any particular scenario. As such, causal estimands must consider averages over data and infer effects using similarities between data points.

To illustrate how such causal principles must be used to estimate quantities of interest we consider a medical scenario: a hospital must decide whether to provide high flow oxygen to all future Covid positive patients presenting to the emergency department ($A=a$), or only use more conservative (and less costly) measures such as antivirals ($A=a'$). The goal is to understand the relationship between these treatment options and possible outcomes of recovery ($Y=y$) or death ($Y=y')$. It is likely that the relationship is confounded by illness severity—patients who receive high flow oxygen are typically very unwell ($X=x$) and those who only receive conservative measures usually have a milder form of the illness ($X=x'$).
Graph \ref{fig:3_node_structure} depicts this scenario. The expected outcome is a function of both the disease severity and the treatment option, and can naively be calculated from the observational data according to 
\begin{align}
\mathbb{E}[Y \mid A=a] = \sum_x \mathbb{E}[Y \mid a,x]~p(x \mid a).
\end{align}

\begin{figure}[!htbp]
    \centering
    \begin{minipage}{0.4\textwidth}
        \centering
        \begin{tikzpicture}
    \node[circle, draw, minimum size=1cm] (A) at (0, 0) {A};
    \node[circle, draw, minimum size=1cm] (X) at (2, 3) {X};
    \node[circle, draw, minimum size=1cm] (Y) at (4, 0) {Y};

    \draw[->, line width=1mm, >=stealth, scale=2] (X) -- (A);
    \draw[->, line width=1mm, >=stealth, scale=2] (X) -- (Y);
    \draw[->, line width=1mm, >=stealth, scale=2] (A) -- (Y);
\end{tikzpicture}
        \caption{\justifying 3 node confounder causal directed acyclic graph. Here $X$ represents the confounders, all random variables that have a causal effect on both the treatment $A$ and the outcome $Y$. There is also a causal link between $A$ and $Y$. These confounder structures frequently result in Simpson's paradox.}
        \label{fig:3_node_structure}
    \end{minipage}%
    \hspace{0.05\textwidth} %
    \begin{minipage}{0.4\textwidth}
        \centering
        \begin{tabular}{|c|c|}
            \hline
            Quantity & Outcome \\
            \hline
            $\Pr(Y = \text{recovery}|A = \text{Antivirals})$ & $0.9$ \\
            $\Pr(Y = \text{recovery}|A = \text{Oxygen})$ & $0.8$ \\
            \hline
        \end{tabular}
        \caption{\justifying Probabilities of outcomes given the treatments $A$. Here we see that the data suggests that antivirals perform better than high flow oxygen. }
        \label{probabilities of outcomes}
    \end{minipage}%
\end{figure}

We see in table \ref{probabilities of outcomes} the corresponding calculated probabilities imply a good strategy might be to prohibit the use of high flow oxygen. However, we can also re-calculate the probabilities of each outcome given such an intervention were to occur. In order to do so, we use the SWIG graphical theory of \cite{Richardson2013SingleWI} that unifies causal directed acyclic graphs (DAGs) and potential (aka counterfactual) outcomes via a node-splitting transformation (see figure \ref{fig:3_node_structure_split}). We represent the intervention by a split node, where $A$ represents the original observed variable, and $a$ represents the interventional situation where $A=a$. In this case we calculate the probability of the expected outcome as
\begin{align}
\mathbb{E}[Y \mid A=a] = \sum_x  \mathbb{E}[Y \mid a,x]~p(x).
\end{align}
From figure \ref{probabilities of outcomes intervention}, we see the adjusted probabilities have reversed in magnitude, a phenomenon known as Simpson's paradox. It is these second set of interventionist expected outcomes that we wish to estimate via measured data on treatments, outcomes and covariates of some population. This particular scenario is the estimation of $\theta_0(a)$, the expected outcome of recovery given all Covid positive patients were given high flow oxygen.

\begin{figure}[!htbp]
    \centering
    \begin{minipage}{0.4\textwidth}
        \centering
        
\begin{tikzpicture}
    \node[half circle, draw, minimum size=0.5cm, rotate=180] (A) at (0, 0.1) {\rotatebox{-180}{A}};
    \node[half circle, draw, minimum size=0.54cm] (a) at (0, 0) {a};
    \node[circle, draw, minimum size=1cm] (X) at (2, 3) {X};
    \node[circle, draw, minimum size=1cm] (Y) at (4, 0) {Y};

    \draw[->, line width=1mm, >=stealth, scale=2] (X) -- (A);
    \draw[->, line width=1mm, >=stealth, scale=2] (X) -- (Y);
    \draw[->, line width=1mm, >=stealth, scale=2] (a) -- (Y);
\end{tikzpicture}
        \caption{\justifying 3 node confounder structure under the intervention $A=a$. We split the node $A$ into two and hence remove the causal link from $X$ to the intervention variable $a$.}
        \label{fig:3_node_structure_split}
    \end{minipage}%
    \hspace{0.05\textwidth} %
    \begin{minipage}{0.4\textwidth}
        \centering
        \begin{tabular}{|c|c|}
            \hline
            Quantity & Outcome \\
            \hline
            $\Pr(Y^{(\text{Antivirals})} = \text{recovery})$ & $0.6$ \\
            $\Pr(Y^{(\text{Oxygen})} = \text{recovery})$ & $0.7$ \\
            \hline
        \end{tabular}
        \caption{\justifying Probabilities of recovery given an intervention were to occur with the treatment specified in the superscript. We see here the probabilities have reversed and the data suggests high flow oxygen is in fact a better intervention. However, also note that both probabilities are lower and suggest that neither intervention is better than maintaining the current treatments. }
        \label{probabilities of outcomes intervention}
    \end{minipage}%
\end{figure}

While the main text considered only dose response, in these appendices we consider the more general family of estimands made up from four causal functions:

\begin{enumerate}
    \item Dose response: $\theta_0(a)=\mathbb{E}\{Y^{(a)}\}$.
     \item Dose response with distribution shift: $ \theta_0^{DS}(a,\tilde{\text{\normalfont P}})=\mathbb{E}_{\tilde{\text{\normalfont P}}}\{Y^{(a)}\}$. 
    \item Conditional response: $ \theta_0^{ATT}(a,a')=\mathbb{E}\{Y^{(a')} \mid A=a\}$.
     \item Heterogeneous response: $\theta_0^{CATE}(a,v)=\mathbb{E}\{Y^{(a)} \mid V=v\}$.
\end{enumerate}
The superscripts (DS, ATT, CATE) denote historically well-known alternative parametric estimators of the same causal quantities: Distribution Shift, Average Treatment on Treated, Conditional Average Treatment Effect. \emph{Dose response with distributional shift} estimates the counterfactual mean outcome given treatment $A=a$ for an alternative population with data distribution $\tilde{\text{\normalfont P}}$. This estimand is valuable in the study of machine learning concepts like transfer learning, distributional shift and covariate shift \cite{datasetshift}. \emph{Conditional response} estimates the counterfactual mean outcome given a treatment $A=a'$ for the sub-population who actually received treatment $A=a$. \emph{Heterogeneous response} estimates the counterfactual mean outcome given treatment $A=a$ for the sub-population with a particular interpretable sub-covariate value $V=v$. This final quantity allows one to ask what causal effects can be expected for individuals in the population with a particular characteristic.

Continuing with our illustrative example, $\theta_0^{DS}(a)$ is the same as $\theta_0(a)$ but instead considers the outcome on a new population, e.g. we have data from Australia but wish to determine the expected outcome of an intervention in America. $\theta_0^{ATT}(a,a')$ is the expected outcome for the population who received the less costly treatment ($a'$) if we instead gave them high flow oxygen ($a$). Lastly, $\theta_0^{CATE}(a,v)$ is the outcome if the patients who have a certain characteristic, such as those over the age of $60$ ($v$), are all given high flow oxygen. The causal structure DAG for CATE is given in figure \ref{fig:CATE DAG}.

\begin{figure}[!htbp]
    \centering
    \begin{subfigure}{0.45\textwidth}
        \centering
        \begin{tikzpicture}
            \node[circle, draw, minimum size=1cm] (A) at (0, 0) {A};
            \node[circle, draw, minimum size=1cm] (X) at (2, 3) {X};
            \node[circle, draw, minimum size=1cm] (Y) at (4, 0) {Y};
            \node[circle, draw, minimum size=1cm] (V) at (4, 3) {V};

            \draw[->, line width=1mm, >=stealth, scale=2] (X) -- (A);
            \draw[->, line width=1mm, >=stealth, scale=2] (X) -- (Y);
            \draw[->, line width=1mm, >=stealth, scale=2] (A) -- (Y);
            \draw[->,line width=0.3mm,dashed] (X) -- (V);
            \draw[->,line width=0.3mm,dashed] (V) -- (X);
            \draw[->, line width=1mm, >=stealth, scale=2] (V) -- (A);
            \draw[->, line width=1mm, >=stealth, scale=2] (V) -- (Y);
        \end{tikzpicture}
        \caption{}
        \label{fig:dag_without_split}
    \end{subfigure}
    \hfill
    \begin{subfigure}{0.45\textwidth}
        \centering
        \begin{tikzpicture}
            \node[half circle, draw, minimum size=0.5cm, rotate=180] (A) at (0, 0.1) {\rotatebox{-180}{A}};
            \node[half circle, draw, minimum size=0.54cm] (a) at (0, 0) {a};
            \node[circle, draw, minimum size=1cm] (X) at (2, 3) {X};
            \node[circle, draw, minimum size=1cm] (Y) at (4, 0) {Y};
            \node[circle, draw, minimum size=1cm] (V) at (4, 3) {V};

            \draw[->, line width=1mm, >=stealth, scale=2] (X) -- (A);
            \draw[->, line width=1mm, >=stealth, scale=2] (X) -- (Y);
            \draw[->, line width=1mm, >=stealth, scale=2] (a) -- (Y);
            \draw[->,line width=0.3mm,dashed] (X) -- (V);
            \draw[->,line width=0.3mm,dashed] (V) -- (X);
            \draw[->, line width=1mm, >=stealth, scale=2] (V) -- (A);
            \draw[->, line width=1mm, >=stealth, scale=2] (V) -- (Y);
        \end{tikzpicture}
        \caption{}
        \label{fig:dag_with_split}
    \end{subfigure}
    \caption{\justifying DAG for CATE showing the relationships between variables X, A, Y, and V  (a) and under the intervention $A=a$ (b). Dashed arrow represents a dependency that is not causal in nature as $V$ is a subset of $X$.}
    \label{fig:CATE DAG}
\end{figure}
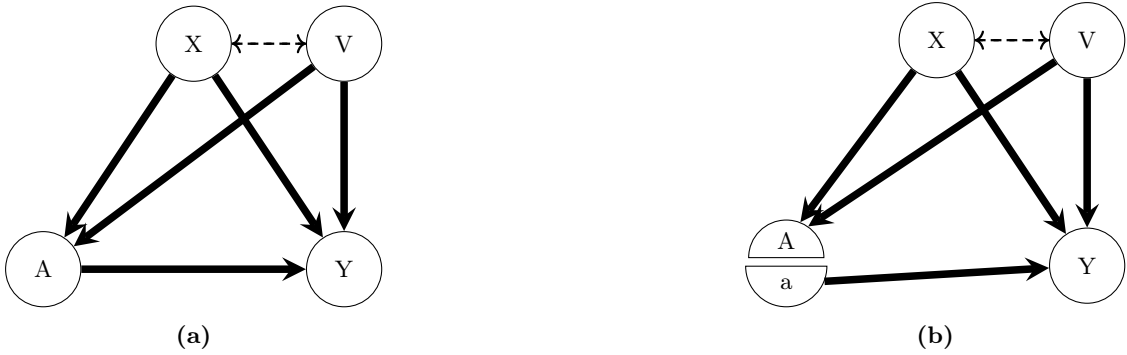

\section{Derivation of Classical Causal Inference Algorithms}\label{Section: Classical Causal Inference Proofs}

Three intuitive \emph{a priori} assumptions underlie causal methods. Here we explain them with further description:

\begin{enumerate}
    \item \emph{No interference}: if $A=a$ then $Y=Y^{(a)}$. That is, the outcome $Y$ for an individual in the population is not affected by the treatments of any other individual in the population.
    \item \emph{Conditional exchangeability}: $\{Y^{(a)}\} \indep A | X$. This implies all important causal factors are included in the covariates $X$, thus if the treatment groups were swapped, we would get the same value for $\theta_0$. %
    \item \emph{Overlap}: if $f(x) > 0$ then $f(a | x) > 0$, where $f(x)$ and $f(a | x)$ are densities. This guarantees that there is no covariate sub-stratum such that treatment has restricted support. Effectively, all treatments are a priori possible for all individuals in the population.  
\end{enumerate}

To estimate dose response with distribution shift, $ \theta_0^{DS}(a,\tilde{\text{\normalfont P}})$, we require two further assumptions:
\begin{enumerate}
\item $\tilde{\text{P}}(Y, A, X)$ = P$(Y | A, X)\tilde{\text{P}}(A, X)$. This implies that even though the marginal distribution of $A$ and $X$ might change in the shifted population, the way $Y$ causally depends on $A$ and $X$ does not. 
\item 
$\tilde{\text{P}}(A, X)$ is absolutely continuous with respect to P$(A, X)$. This implies the set of possible values for $\tilde{P}$ is contained within $P$, i.e. we can not observe $X$ and $A$ in the new distribution that were impossible in the prior distribution.
\end{enumerate}

Given the above assumptions, one can write the four causal functions as integrals:
\begin{align}
    \theta_0(a)  &= \int E(Y | A = a, X = x) d \text{P}(x) \label{ATE eq 1}\\
    \theta^{DS}_0(a) &= \int E(Y | A = a, X = x) d \tilde{\text{P}}(x)\\
    \theta^{ATT}_0(a,a')  &= \int E(Y | A = a, X = x) d \text{P}(x|a') \\
    \theta^{CATE}_0(a,v) &= \int E(Y | A = a, V=v, X = x) d \text{P}(x|v), \label{CATE eq 1}
\end{align}    
where $\text{P}(x)$ is the probability distribution for the covariates $x \in X$. For brevity we denote $\gamma_0(a,x) = E(Y|A=a,X=x)$, and $\gamma_0(a,v,x) = E(Y|A=a,V=v,X=x)$.

While many causal estimators make parametric assumptions on the form of the non-linear regression function $\gamma_0(a,v,x)$, and conditional function $\gamma_0(a,x)$,  these assumptions can be relaxed by using a Reproducing Kernel Hilbert Space (RKHS) as a hypothesis space. An attractive feature of doing so is that one can then learn $\gamma_0$ using a kernel ridge regression (KRR) estimator to derive closed form solutions for all four estimands with uniform consistency guarantees \cite{grettonkernelcausal}. 

We now introduce the basics of RKHS theory and direct the reader to  \cite{grettonkernelcausal, schoelkopf2002learning} for further detail. A RKHS is characterised by its feature map which takes a point $x$ in the original data space $X$ (which need not take any particular form and can include spaces of graphs, images, text etc) and maps it to a feature $\phi(x)$ in the RKHS $\mathcal{H}$. The closure of the span$\{\phi(x)\}$ defines the RKHS, such that the $\{\phi(x)\}_{x \in \mathcal{X}}$ form basis functions for the RKHS. The \emph{kernel} is then defined as the symmetric, continuous, positive definite function $k(x,x') = \langle \phi(x), \phi(x')\rangle_\mathcal{H}$. A key property of an RKHS is the reproducing property: to evaluate a function $\gamma_0$ at $x$, one can take the inner product between $\gamma_0$ and $\phi(x)$.
\begin{align}
\gamma_0(x) = \langle \gamma_0, \phi(x) \rangle
\end{align}
Finally, due to the representer theorem, the function we are trying to learn will lie in the span of the training data: $ \gamma_0 = \sum_{i=1}^n \alpha_i\phi(x_i)$. Given these facts, a KRR estimator for $\gamma_0$ can be written in terms of an RKHS inner product:

\begin{equation}\label{eq:cef_loss}
\hat{\gamma_0}=\argmin_{\gamma_0 \in\mathcal{H}} {\left[n^{-1}\sum_{i=1}^n \{Y_i-\langle\gamma_0,\phi(X_i)\rangle_{\mathcal{H}}\}^2 + \lambda \|\gamma_0\|^2_{\mathcal{H}}\right]},
\end{equation}
where $\lambda>0$ is a hyperparameter that tunes the ridge penalty $\|\gamma_0\|^2_{\mathcal{H}}$ and promotes smoothness in estimation. Furthermore, the solution to this optimization problem can be written in closed form \cite{kimeldorf1971some}:
\begin{equation}\label{eq:cef_form}
\hat{\gamma}(x)=Y^{\top}(K_{XX}+n\lambda  I )^{-1}K_{Xx},
\end{equation}
where $K_{XX}\in\mathbb{R}^{n\times n}$ is the kernel matrix, with $(i,j)$th entry $k(x_i,x_j)$, and $K_{Xx}\in\mathbb{R}^n$ is the kernel vector, with $i$th entry $k(x_i,x)$.

Following from equations \ref{ATE eq 1} - \ref{CATE eq 1}, we specify two further assumptions related to the existence of a lernable regression $\gamma_0$ and choice of kernel to yield a non-parametric equation to estimate $\theta_0$. 

First, to use kernel methods we must assume that the expectation values of outcomes, given specific treatments are learnable regressions. Formally, this means they exist in some Reproducing Kernel Hilbert Space given by the product of decoupled kernels.

\begin{assumption}\label{learnable} (Learnable regression $\gamma_0$)\\
We assume that $\gamma_0(a,x)$ is an element of the RKHS with kernel $k(a,x;a',x') = k_D(a,a')k_X(x,x')$, giving the tensor product RKHS $\mathcal{H} = \mathcal{H}_D \otimes \mathcal{H}_X$. And similarly $\gamma_0(a,v,x) \in \mathcal{H}_D \otimes \mathcal{H}_V \otimes \mathcal{H}_X$.  
\end{assumption}

This allows us to write $\gamma_0(d, x) = \bra{\gamma_0}\ket{\phi(a)\otimes \phi(x)}_{\mathcal{H}}$ by the reproducing property. And similarly $\gamma_0(a,v, x) = \bra{\gamma_0}\ket{\phi(a)\otimes\phi(v)\otimes \phi(x)}_{\mathcal{H}}$. This decouples the feature maps of the components $X,D$ (and $V$ for CATE).

We next must assume properties of such kernels so that we can accurately use them to estimate $\theta_0$.

\begin{assumption}\label{regularity} (RKHS regularity assumptions)
   \begin{enumerate}
       \item  $k_\mathcal{A}, k_\mathcal{V}, k_\mathcal{X}$ (and $k_\mathcal{Y}$) are continuous and bounded. Formally, $\sup_{a\in\mathcal{A}} \|\phi(a)\|_{\mathcal{H}_\mathcal{A}} \leq m_a$,
  $\sup_{v\in\mathcal{V}} \|\phi(v)\|_{\mathcal{H}_\mathcal{V}} \leq m_v$, $\sup_{x\in\mathcal{X}} \|\phi(x)\|_{\mathcal{H}_\mathcal{X}} \leq m_x$ {and $\sup_{y\in\mathcal{Y}} \|\phi(y)\|_{\mathcal{H}_\mathcal{Y}} \leq m_y$}. This is for Bochner integrability and the existence of mean embedding.
\item $\phi(a), \phi(v), \phi(x)$ {and $\phi(y)$} are measurable. 
\item $k_\mathcal{X}$ (and $k_\mathcal{Y}$) are characteristic. This is so that mean embeddings are injective.
   \end{enumerate}
\end{assumption}

Assumptions \ref{regularity}.1 and \ref{regularity}.2 are weak assumptions on kernels and are satisfied for all common choices of kernels. Assumption \ref{regularity}.3 is more stringent and requires more careful choice of kernel. We note that this is the only assumption on the form or properties of the kernel, all other assumptions surround the problem formulation and the true distributions we are trying to learn. Practically this implies mild constraints on the choice of kernel function and instead suggests that the kernel form should be designed to match that of the data structure or any apriori knowledge of the problem.

Given all prior assumptions, we arrive at lemma \ref{kernalising response} where we can write the counterfactual mean response as a kernel value including the mean feature embeddings - guaranteed to be injective via the characteristic property of $k_\mathcal{X}$. 

\begin{lemma}\label{kernalising response}(Kernalising Estimands)\\
\begin{align}
    \theta_0(a) &= \int \mathbb{E}(Y | A = a, X = x) d \text{P}(x) \nonumber \\
    &=\bra{\gamma_0}\ket{\phi(a)\otimes \mu_x}_{\mathcal{H}}\\
    \theta_0^{DS}(a) &= \bra{\gamma_0}\ket{\phi(a)\otimes \nu_x}_{\mathcal{H}}\\
    \theta_0^{ATT}(a,a') &= \bra{\gamma_0}\ket{\phi(a) \otimes \mu_x(a')}_{\mathcal{H}}\\
    \theta_0^{CATE}(a,v) &= \bra{\gamma_0}\ket{\phi(a)\otimes \phi(v) \otimes \mu_x(v)}_{\mathcal{H}},
\end{align}
where $\mu_x = \int \phi(x) d \text{P}(x),\nu_x = \int \phi(x) \tilde{\text{P}}(x) $ and $\mu_x(\cdot) = \int \phi(x) d \text{P}(x|\cdot)$.
\end{lemma}

Below we present a brief proof for $ATE$ with some explanation, a further explanation and proofs for all four functions are given in \cite{grettonkernelcausal}.

\begin{align}
    \theta_0(a) &= \int E(Y | A=a, X = x) d \text{pr}(x) \label{ATE proof eq1}\\
    &= \int \gamma_0(a,x) d \text{pr}(x) \quad \text{(Definition)}\\
    &= \int \bra{\gamma_0}\ket{\phi(a)\otimes \phi(x)}_{\mathcal{H}} d \text{pr}(x) \quad \text{(Reproducing Property)}\\
    &= \bra{\gamma_0}\ket{\int \phi(a)\otimes \phi(x) d \text{pr}(x)}_{\mathcal{H}} \quad \text{(Bochner Integrability)}\\
    &=\bra{\gamma_0}\ket{ \phi(a)\otimes \int \phi(x) d \text{pr}(x)}_{\mathcal{H}}  \quad (\phi(a) \indep \text{pr}(x))\\
    &=\bra{\gamma_0}\ket{\phi(a)\otimes \mu_x}\label{ATE proof eq final}\\
\end{align}
If we now approximate the mean embeddings by the average feature maps over measured data or via KRR,

\begin{align}
    \hat{\mu}_x &= \frac{1}{n} \sum_{x_i} \phi(x_i) \\
    \hat{\nu}_x &= \frac{1}{\tilde{n}} \sum_{\tilde{x_i}} \phi(\tilde{x_i}) \\
    \hat{\mu}_x(i)(\cdot) &= K_{(\cdot)X}(K_{II} + n \lambda_2 I)^{-1}K_{Ii},
\end{align}
we can estimate the counterfactual response as shown in lemma \ref{estimation}, arriving at a closed form non-parametric classical algorithm to estimate the counterfactual mean distributions.

\begin{lemma}\label{estimation}(Causal Function Estimation)\\
We can estimate counterfactual response via kernel ridge regression as,
    \begin{align}
    \hat{\theta}(a) &= \bra{\gamma_0}\ket{\phi(a)\otimes \hat{\mu}_x}\\
    &= n^{-1} Y^T (K_{AA} \odot K_{XX} +  n\lambda I)^{-1} (K_{Aa} \odot \sum_{x_i} K_{Xx_i})\\
    \hat{\theta}^{DS}(a) &= \tilde{n}^{-1} Y^T (K_{AA} \odot K_{XX} +  n\lambda I)^{-1} (K_{Aa} \odot \sum_{\tilde{x_i}} K_{X\tilde{x_i}}) \\
    \hat{\theta}_0^{ATT}(a,a') &= Y^T \qty(K_{AA} \odot K_{XX} + n\lambda I)^{-1}\qty[K_{Aa} \odot \{K_{XX}(K_{AA} + n\lambda_{1}I)^{-1}K_{Aa'}\}]\\
    \hat{\theta}_0^{CATE}(a,v) &= Y^T \qty(K_{AA} \odot K_{VV} \odot K_{XX} + n\lambda I)^{-1} \qty[K_{Aa} \odot K_{Vv} \odot \{K_{XX}(K_{VV} + n\lambda_{2}I)^{-1}K_{Vv}\}]
\end{align}
\end{lemma}
Sketch of proof for $ATE$, further details in \cite{grettonkernelcausal},

\begin{align}
    \hat{\theta}(a) &= \bra{\gamma_0}\ket{\phi(a)\otimes \hat{\mu}_x}\\
    &= \bra{\gamma_0}\ket{\phi(a)\otimes n^{-1} \sum_{x_i} \phi(x_i)} \\
    &= n^{-1}\sum_{x_i} \bra{\gamma_0}\ket{\phi(a)\otimes \phi(x_i)} \quad \text{(Linearity of Tensor Prod)}\\
    &= n^{-1} Y^T (K_{AA} \odot K_{XX} + n\lambda I)^{-1} (K_{Aa} \odot \sum_{x_i} K_{Xx_i}) \quad \text{standard KRR}.
\end{align}

\section{Proof of Uniform Consistency for Classical Causal Inference Algorithm}\label{Section: Classical Uniform Consistency Results}

In applications with sensitive outcomes such as healthcare and econometrics, it is especially relevant to understand whether or not algorithms are accurate---interventions can have harmful, real world consequences if estimated incorrectly. It is therefore important to bound the error between our causal estimand and the true causal effect across the full range of possible interventions. Statistical estimators can be described as weakly consistent, strongly consistent and uniformly consistent. When working with very large populations, one can consider asymptotic behavior: a \emph{weakly consistent} estimator is one which for any $\epsilon$, the probability of an estimator being $\epsilon$ close to the true quantity approaches 1 \cite{Vaart_1998}. Formally an estimator $\hat{\theta}_n$ is weakly consistent if for any $\epsilon$,

\begin{equation}
    \lim_{n\to\infty} P(|\hat{\theta}_n - \theta| < \epsilon) = 1,
\end{equation}
where $\theta$ is the true value we wish to estimate. This means that while each individual estimator may not be $\epsilon$ close to the true value, if we generate many random samples, the probability that the estimator is close to $\theta$ approaches 1. \emph{Strong consistency} is a stronger condition, where if we generate many estimators from random samples, \textit{every} estimator must converge to $\theta$ as our sample size approaches $\infty$. Formally this is written as,

\begin{equation}
    P(\lim_{n\to \infty} \hat{\theta}_n = \theta) = 1.
\end{equation}

Strong consistency implies weak consistency. Intuitively, weak consistency allows for edge case errors to appear, as long as the probability of them appearing tends to $0$. Strong consistency specifies that the estimator \textit{must} converge and hence edge case errors cannot appear after some large enough $n$. 

From here on, we use consistency and strong consistency interchangeably unless otherwise specified. Traditionally in statistics, one only considers estimating a singular value, usually for a hypothesis test, however, in causal inference, we often estimate response curves \cite{10.1093/biomet/90.3.491}. Thus, we wish the estimator to converge in probability to the true underlying function, uniformly over the entire treatment domain, as the sample size approaches infinity. Point-wise strong consistency implies that a function must be strongly consistent for any input $x$. Formally for any $x\in \mathcal{X}$, 

\begin{equation}
    P(\lim_{n\to \infty} \hat{\theta}_n(x) = \theta(x)) = 1.
\end{equation}

However, if one wishes to consider worst case error scenarios - which is useful in sensitive applications - a stronger condition is required. \emph{Uniform consistency} ensures that every $\delta$ neighborhood of each point must converge. Uniform strong consistency can be written in terms of supremum norm (i.e. worst case error),

\begin{equation}
    P(\lim_{n\to \infty} \sup_{x} | \hat{\theta}_n(x) - \theta(x)| = 0) = 1.
\end{equation}
This is the definition we work with for the remainder of the appendix. Note that uniform consistency does not consider the \emph{rate} at which an estimator will converge. To quantify how quickly the estimator converges to the true function as the sample size increases, further analysis is required.

To prove uniform consistency of the causal functions one must make two assumptions on the form of the expectation function, $\gamma_0$. This proof has been adapted and expanded from \cite{grettonkernelcausal}. Recall, $\gamma_0$ is an element of an RKHS, which by the representer theorem one can write as

\begin{equation}
    f = \sum_{i=0}^\infty \alpha_i \phi_i,
\end{equation}
where $\alpha_i \in \mathbb{C}$ and $\phi_i$ are the basis functions of the RKHS $\mathcal{H}$. Any function in an RKHS with a bounded kernel is square integrable and hence we can say that for any $f \in \mathcal{H}$,

\begin{equation}
    \sum_{i=0}^\infty \frac{\alpha_i^2}{\eta_i} < \infty,
\end{equation}
where $\eta_i$ are the eigenvalues of $\phi_i$, with the eigenvalues ordered from largest to smallest, i.e. $\eta_i \geq \eta_{i+1}$. To define \textit{additional} smoothness, one can introduce a constant $c > 1$, and assume that the regression, $\gamma_0$, satisfies,

\begin{equation}
    \gamma_0 \in \mathcal{H}^c := \{f = \sum_{i=0}^\infty \gamma_i \phi_i | \sum_{j=0}^\infty \frac{\gamma_j^2}{\eta_j^c} < \infty \}. 
\end{equation}

This states that $\gamma_0$ is not only square integrable but has additional smoothness, determined by the value of $c$, such that $\gamma_0$ is well approximated by leading terms of $\phi_i$. One considers values $c\in(1,2]$, where values of $c$ closer to $c=2$ generate a smoother target $\gamma_0$ and thus faster convergence for an estimator $\hat{\gamma}$. One can also define similar properties for the mean embeddings, $\mu_x(a), \mu_x(v)$ for conditional response and heterogeneous response estimators via the constants $c_1, c_2$, respectively.

 The second assumption one can make to guarantee rapid convergence, is a  restricted RKHS effective dimension. While the above assumption implies that the fraction $\frac{\gamma_j^2}{\eta_j^c}$ vanishes, here one places a direct assumption on how fast the eigenvalues $\eta_j$ decay, which emphasises the early $\phi_j$ basis functions. Mathematically, we assume there exists some constant $C$ such that,

\begin{equation}
    \eta_j \leq C j^{-b}.
\end{equation}

Bounded kernels imply that $b \geq 1$ \cite{sobolev}. Intuitively, higher values of $b$ give greater importance to earlier values of $\phi_i$ and hence corresponds to a lower effective dimension. This leads to faster convergence of an estimator $\hat{\gamma_0}$, as $\gamma_0$ can now be approximated well by fewer terms. The limit $b\to \infty$ can essentially be regarded as a finite dimensional RKHS \cite{caponnetto}. For conditional response and heterogeneous response mean embeddings $\mu_x(a)$ and $\mu_x(v)$ one can similarly define values $b_1$ and $b_2$ for the respective RKHS's to enforce fast convergence.

It has been shown that under Sobolov norm assumptions the minimax optimal rate of the causal estimators is $n^{-\frac{c-1}{2(c+1/b)}}$, a result which leverages kernel ridge regression and uses the fact that optimal regularisation occurs when $\lambda = n^{-1/(c+1/b)}$ \cite{sobolev}. Gretton et, al. \cite{grettonkernelcausal} have recently utilised these results to prove that the four classical causal algorithms are uniformly consistent and converge for large $n$, with probability $1-2\delta$, at the following rates,

\begin{align}
    \|\hat{\theta} - \theta_0\|_{\infty} &= \mathcal{O} \left[\log(1/\delta) n^{-(c-1)/\{2(c+1/b)\}}\right],\label{unif consistent equation ATE}\\
    \|\hat{\theta}_0^{DS}(\cdot, \tilde{\mathbf{p}}) - \theta_0^{DS}(\cdot, \tilde{\mathbf{p}})\|_{\infty} &= \mathcal{O}\left[ \log(1/\delta) \qty(n^{-(c-1)/\{2(c+1/b)\}} + \tilde{n}^{-1/2})\right],\\
    \|\hat{\theta}^{ATT} - \theta_0^{ATT}\|_{\infty} &= \mathcal{O}\left[\log(1/\delta)\qty(n^{-(c-1)/\{2(c+1/b)\}} + n^{-(c_1-1)/\{2(c_1+1/b_1)\}})\right],\\
    \|\hat{\theta}^{CATE} - \theta_0^{CATE}\|_{\infty} &= \mathcal{O}\left[\log(1/\delta)\qty(n^{-(c-1)/\{2(c+1/b)\}} + n^{-(c_2-1)/\{2(c_2+1/b_2)\}})\right].
\end{align}

Furthermore, when the effective dimension is finite $(b,b_1,b_2 \to \infty)$ and the targets $\gamma_0, \mu_x(a), \mu_x(v)$  are maximally smooth $(c,c_1,c_2=2)$, all estimands will converge in $\mathcal{O}(n^{-1/4})$. Note, this final rate reflects a ``best case'' result, which in many real world situations may not be realisable. Additionally, optimal $\lambda$ is not known \emph{a priori} and is usually learned through cross validation techniques.

Below we sketch a proof of the uniform consistency for $\theta_0$ taken from \cite{grettonkernelcausal}. The first few steps are adding and subtracting certain terms to get differences between estimates and true values in each entry of the kernel,

\begin{align}
&\hat{\theta}(a) - \theta_0(a) = \langle \hat{\gamma}, \phi(a) \otimes \hat{\mu}_x \rangle_\mathcal{H} - \langle \gamma_0, \phi(a)  \otimes \mu_x \rangle_\mathcal{H} \\
&= \langle \hat{\gamma}, \phi(a)  \otimes (\hat{\mu}_x - \mu_x) \rangle_\mathcal{H} + \langle (\hat{\gamma} - \gamma_0), \phi(a)  \otimes \mu_x \rangle_\mathcal{H} \\
&= \langle (\hat{\gamma} - \gamma_0), \phi(a)  \otimes (\hat{\mu}_x - \mu_x) \rangle_\mathcal{H} + \langle \gamma_0, \phi(a)  \otimes (\hat{\mu}_x - \mu_x) \rangle_\mathcal{H} + \langle (\hat{\gamma} - \gamma_0), \phi(a)  \otimes \mu_x \rangle_\mathcal{H}.
\end{align}
Therefore by the Cauchy-Schwarz inequality bounds for our estimators in $\gamma$ and $\mu_x$ we have with probability $1 - 2\delta$,

\begin{align}
|\hat{\theta}(a) - \theta_0(a)| &\leq \|\hat{\gamma} - \gamma_0\|_\mathcal{H} \|\phi(a)\|_{\mathcal{H}_\mathcal{A}} \|\hat{\mu}_x - \mu_x\|_{\mathcal{H}_x} \\
& + \|\gamma_0\|_\mathcal{H} \|\phi(a)\|_{\mathcal{H}_\mathcal{A}} \|\hat{\mu}_x - \mu_x\|_{\mathcal{H}_x} \\
&+ \|\hat{\gamma} - \gamma_0\|_\mathcal{H} \|\phi(a)\|_{\mathcal{H}_\mathcal{A}} \|\mu_x\|_{\mathcal{H}_x} \\
&\leq m_a \cdot r_\gamma(n, \delta, b, c) \cdot r_\mu(n, \delta) + m_a \|\gamma_0\|_\mathcal{H} \cdot r_\mu(n, \delta) + m_a m_x \cdot r_\gamma(n, \delta, b, c) \\
&= \mathcal{O} \left(n^{-\frac{1}{2}\cdot \frac{c}{c + \frac{1}{b}}}\right).
\end{align}

\section{Uniform Consistency Guarantee and Convergence Rates with Quantum Dose Response, Distribution Shift, Heterogenous Response and Conditional Response}\label{Section: Estimation Convergence with Quantum Algorithms}

Following from the main text, we can provide equivalent proofs for the uniform consistency guarantees and convergence rates of our other three quantum algorithms (distribution shift, heterogenous response and conditional response). 

Identically to dose response ($\theta(a)$), we can note that for distribution shift ($\theta^{DS}(a,\tilde{P})$),

\begin{align}
    |\hat{\theta}_{quantum}^{DS}(a) - \theta_0^{DS}(a)| &\leq|\hat{\theta}_{quantum}^{DS}(a) - \hat{\theta}^{DS}(a)| + |\hat{\theta}^{DS}(a) - \theta_0^{DS}(a)| \\
    &\leq\epsilon_k +  \norm{Y} \norm{A^{-1}b} \epsilon + \mathcal{O}(n^{-\frac{1}{2}\frac{c-1}{c+1/b}} + \tilde{n}^{-\frac{1}{2}}),
\end{align}
again permitting the same convergence rate when considering $\epsilon$ but holding the same tradeoff for $\epsilon_k$.

Next we look at the consistency for conditional response ($\hat{\theta}_0^{CATE}(a,v)$) and heterogeneous response $(\hat{\theta}_0^{ATT}(a,a'))$. These algorithms can be similarly evaluated via quantum subroutines as shown in algorithm \ref{alg 2}.

\begin{algorithm}[!htbp]
\caption{Quantum Conditional Dose Response $\hat{\theta}_0^{CATE}(a,v)$\,\,($\hat{\theta}_0^{ATT}(a,a')$)}
\label{alg 2}
\KwIn{Kernel entries: $K_{AA}, K_{VV}, K_{XX}, K_{Aa}, K_{Vv}$\,\,($K_{Aa'}$); regularization parameters $\lambda, \lambda_2$\,\,($\lambda_1$); number of data points $n$; outcome vector $Y$.}
\KwOut{Approximate quantum state with amplitude proportional to $\hat{\theta}_0^{CATE}(a,v)$\,\,($\hat{\theta}_0^{ATT}(a,a')$)}

\textbf{Classical Computation:} \\
Compute $A_1 = K_{AA} \odot K_{VV} \odot K_{XX} + n\lambda I$\,\,($K_{AA} \odot K_{XX} + n\lambda I$)\;
Compute $A_2 = K_{VV} + n \lambda_2 I$\,\,($K_{AA} + n \lambda_1 I$)\;
Normalize $b = K_{Vv}$\,\,($K_{Aa}$)\;

\textbf{Quantum Subroutine:} \\
Prepare quantum states $\ket{b}$ and $\ket{Y}$\;
Use a quantum linear solver to obtain $\ket{A_2^{-1} b}$ to additive error $\epsilon_2$\;
Apply the matrix $K_{XX}$ via a block encoding to $\ket{A_2^{-1} b}$\;
Compute Hadamard product:\\
$\ket{x} = K_{Aa} \odot K_{Vv} \odot \ket{K_{XX} A_2^{-1} b}$\,\,($K_{Aa'} \odot \ket{K_{XX} A_2^{-1} b}$)\;
Apply a quantum linear solver to obtain $\ket{A_1^{-1} x}$ to additive error $\epsilon_1$\;
Measure overlap $\bra{Y}\ket{A_1^{-1} x}$\;
\end{algorithm}

As the quantum algorithms for Conditional Response and Heterogeneous Response each require two adiabatic inversions, with additive errors $\epsilon_1,\epsilon_2$, we can apply a similar process as above. Beginning with ATT, we note that

\begin{align}
    |\hat{\theta}_{quantum}^{ATT}(a,a') - \theta_0^{ATT}(a,a')| &= |\hat{\theta}_{quantum}^{ATT}(a,a') - \hat{\theta}^{ATT}(a,a') + \hat{\theta}^{ATT}(a,a') - \theta_0^{ATT}(a,a')|   \\
    &\leq|\hat{\theta}_{quantum}^{ATT}(a,a') - \hat{\theta}^{ATT}(a,a')| + |\hat{\theta}^{ATT}(a,a') - \theta_0^{ATT}(a,a')| \\
    &\leq |\bra{Y}\ket{A_2^{-1}x}_{quantum} - \langle{Y},{A_2^{-1}x}\rangle| + \mathcal{O}\qty(n^{-\frac{1}{2}\frac{c-1}{c+1/b}} + n^{-\frac{1}{2}\frac{c_1-1}{c_1+1/b_1}}).\\
    &\leq |\bra{Y}\ket{A_2^{-1}x}_{quantum} - \langle{Y},{A_2^{-1}x}\rangle_{quantum}| + |\langle{Y,{A_2^{-1}x}\rangle_{quantum} - \langle Y, {A_2^{-1}x}} \rangle| \nonumber\\ 
    &+ \mathcal{O}\qty(n^{-\frac{1}{2}\frac{c-1}{c+1/b}} + n^{-\frac{1}{2}\frac{c_1-1}{c_1+1/b_1}})\\
    &\leq \epsilon_k + \norm{{Y}}\norm{\ket{A_2^{-1}x}_{quantum} - \ket{A_2^{-1}x}} + \mathcal{O}\qty(n^{-\frac{1}{2}\frac{c-1}{c+1/b}} + n^{-\frac{1}{2}\frac{c_1-1}{c_1+1/b_1}}).\label{final ATT}
\end{align}

We hence need to find a bound for $\norm{\ket{A_2^{-1}x}_{quantum} - \ket{A_2^{-1}x}}$ where $\ket{x} = \ket{K_{Aa'}} \odot \ket{K_{XX} A_1^{-1}b}$. We can bound the error in $x$ by considering,

\begin{align}
    \norm{\ket{A_1^{-1}b}_{quantum} - \ket{{A_1^{-1}b}}} &\leq \norm{A_1^{-1}b} \epsilon_1\\
    \norm{\ket{K_{XX} A_1^{-1}b}_{quantum} - \ket{K_{XX} {A_1^{-1}b}}} &\leq \norm{K_{XX}} \norm{A_1^{-1}b} \epsilon_1\\
    \norm{\ket{K_{Aa'}} \odot \ket{K_{XX} A^{-1}b}_{quantum} - \ket{K_{Aa'}} \odot \ket{K_{XX} {A^{-1}b}}} &\leq \sup_{i,j} \{K_{Aa'}^{(i,j)}\} \norm{K_{XX}} \norm{A_1^{-1}b} \epsilon_1\\
    \norm{\ket{\hat{x}} - \ket{x}} &\leq m_a \norm{K_{XX}} \norm{A_1^{-1}b} \epsilon_1. 
\end{align}
Here we denote $\hat{x}$ as the quantum approximation of the vector $\ket{x}$ by quantum inverse, matrix multiplication and hadamard product. The evaluation of the second matrix inversion on the quantum approximation state $\ket{x}$ then yields the inequality,

\begin{align}
    \norm{\ket{A_2^{-1} \hat{x}}_{quantum} - \ket{A_2^{-1}x}} &\leq \norm{\ket{A_2^{-1} \hat{x}}_{quantum} - \ket{A_2^{-1}\hat{x}}} + \norm{\ket{A_2^{-1} \hat{x}} - \ket{A_2^{-1}x}}\\
    &\leq \norm{A_2^{-1}\hat{x}}\epsilon_2 + \norm{A_2^{-1}} \norm{ \ket{\hat{x}} - \ket{x}}\\
    &\leq \norm{A_2^{-1}\hat{x}}\epsilon_2 + \norm{A_2^{-1}} m_a \norm{K_{XX}} \norm{A_1^{-1}b} \epsilon_1.
\end{align}
We can substitute this into equation \ref{final ATT} to find that, 

\begin{align}
    |\hat{\theta}_{quantum}^{ATT}(a,a') - \theta_0^{ATT}(a,a')| &\leq \epsilon_k +  \norm{Y} \norm{A_2^{-1}\hat{x}}\epsilon_2 +\norm{Y} \norm{A_2^{-1}} m_a \norm{K_{XX}} \norm{A_1^{-1}b} \epsilon_1 \nonumber\\ &+ \mathcal{O}\qty(n^{-\frac{1}{2}\frac{c-1}{c+1/b}} + n^{-\frac{1}{2}\frac{c_1-1}{c_1+1/b_1}}).
\end{align}
Importantly, $\norm{Y}, \norm{A_2^{-1}x}, \norm{A_2^{-1}}, m_a, \norm{K_{XX}}$ and $\norm{A_1^{-1}b}$ are all bounded and independent of $a$.
Again we can see that to maintain the same convergence rate of the classical algorithm, we must have $\epsilon_1,\epsilon_2$ such that 

\begin{equation}
    \epsilon_1,\epsilon_2 \in \Omega(e^{-\Polylog (n)}) \cap \mathcal{O}\qty(n^{-\frac{1}{2}\frac{c-1}{c+1/b}} + n^{-\frac{1}{2}\frac{c_1-1}{c_1+1/b_1}}),
\end{equation}
if the output of our algorithm is a quantum state rather than the classical estimate. However, if we want the classical output we must include the measurement error $\epsilon_k$ which scales as $\Omega(1/\Polylog (n)) > \mathcal{O}\qty(n^{-\frac{1}{2}\frac{c-1}{c+1/b}} + n^{-\frac{1}{2}\frac{c_1-1}{c_1+1/b_1}})$. Hence for optimal convergence while retaining exponential speedup for the quantum inverses we require,

\begin{align}
    |\hat{\theta}_{quantum}^{ATT}(a,a') - \theta_0^{ATT}(a,a')| &\leq \epsilon_k \in \Theta(1/\Polylog (n)),
\end{align}
which does converge to $0$, hence retaining uniform consistency but does so with an exponentially slower convergence rate.

The same logic applies to CATE, except we have an additional $m_v$ factor,

\begin{align}
    |\hat{\theta}_{quantum}^{CATE}(a,v) - \theta_0^{ATT}(a,v)| &\leq \epsilon_k + \norm{Y} \norm{A_2^{-1}\hat{x}}\epsilon_2 +\norm{Y} \norm{A_2^{-1}} m_a m_v \norm{K_{XX}} \norm{A_1^{-1}b} \epsilon_1 \nonumber \\ &+ \mathcal{O}\qty(n^{-\frac{1}{2}\frac{c-1}{c+1/b}} + n^{-\frac{1}{2}\frac{c_2-1}{c_2+1/b_2}}),
\end{align}
which does not change the conditions for $\epsilon_1,\epsilon_2, \epsilon_k$ as constant factors are absorbed in $\mathcal{O}, \Omega, \Theta$ notation. 

Practically, this implies that we can apply the adiabatic inverse accurately enough to maintain a quantum advantage while not reducing the asymptotic scaling guarantee of any of the four classical algorithms. However, any overlap evaluation which requires $1/\epsilon_k^2$ runs forces a tradeoff between tight convergence bounds and the exponential speedup of our quantum algorithm. 

Despite this, even with additive and sampling errors we retain uniform consistency of the quantum algorithm--- a  highly desirable result in causal estimation applications for healthcare and econometrics.

\section{Background on Quantum Kernel Evaluation}\label{Section: Quantum kernel evaluation}

Recall, one of the computational bottlenecks of the classical algorithms is the construction of the kernel matrices. While one cannot improve on this complexity using quantum resources, it has been conjectured that using an explicitly quantum feature embedding to generate the specific kernel for evaluation may offer some advantage in machine learning tasks. Classical kernel evaluation typically makes use of the ``kernel trick'', where one uses an analytic expression to evaluate the kernel $k(x,x')$, rather than explicitly calculating the inner product in feature space $\langle \phi(x), \phi(x')\rangle$. However,  one can also evaluate kernel entries using quantum hardware, where the kernel corresponds to the overlap between two quantum states. Each individual data example is encoded in a quantum state, where this encoding effectively transforms the data to a higher dimensional RKHS. The kernel inner product evaluation can then be estimated through measurement. Each of the kernels in the algorithm above has the potential to be evaluated using such methods, although it is not yet clear if there is an advantage in doing so.

While it seems plausible that quantum kernels may offer some advantage over classical kernels given their access to a larger hypothesis space of learned functions, to date very few learning tasks have been identified where there is a provable advantage. The use of quantum kernels to solve the discrete logarithm problem is a notable exception \cite{Liu_2021} and is continuing to inspire work in this direction  \cite{schuld2021-quantumkernel,Gil-Fuster:2023ifl,schuld_PhysRevLett.122.040504,Hubregtsen_2022,Gyurik_2023,Jerbi_2023,thanasilp2022exponential,Slattery_2023,glick2022covariant,Ghobadi_PhysRevA.104.052403,Bowie_2023,Tiwari_2022_coherent,thanasilp2022exponential,Huang_2021,kübler2021inductive,Enhancing_2024}. Furthermore, while much has been said about generalisability and learnability in the context of quantum kernels, both these properties are functions of the choice of kernel, the training data, and any hyperparameter tuning that may occur during training \cite{canatar2023bandwidth,Shaydulin_2022}. As such, it is likely that any statements about learnability and generalisability would need to consider the specific learning task in question. 

In order to maintain the four proposed quantum algorithms, we require injectivity of mean embeddings and hence any quantum kernel encoding must also be \emph{characteristic}. Characteristic kernels ensure that mapping of probability distributions to a point (i.e. function of $x$) in the RKHS can occur without loss of information. While characteristic quantum kernels are known to exist \cite{henderson2024quantumkernelmachinelearning,schuld2021-quantumkernel}, not all quantum kernels will satisfy this criteria.

\section{Incremental Functions, Counterfactual Distributions and Graphical Models}\label{Incremental Functions, Counterfactual Distributions and Graphical Models}

In addition to the problem classes previously discussed --— such as dose-response estimation, heterogeneous treatment effects, conditional response, and distributional shift --— the appendices of \cite{grettonkernelcausal} present ten further problem structures that provide a more granular breakdown of causal inference settings. These include four incremental problem structures, which focus on the derivative interventions; four counterfactual structures, which centre on estimating the whole distribution under hypothetical interventions; and two graphical problem structures, which leverage known or assumed causal graphs to constrain or inform inference even with unobserved confounders. Each of these problem classes offers a different lens on the causal learning task and poses unique challenges for model design. A quantum implementation of these algorithms naturally follows the previous derivation with different constant factors, namely the supremum of the kernel matrix. For completeness we include a description of each of these causal estimands and show the difference in constant factors.

\subsection{Incremental Functions}\label{incremental functions}

One may be interested to explore incremental functions, defined as the hypothetical outcome of the derivative of an intervention. This encodes how small changes in the intervention procedure changes the outcome distribution.   

The four incremental functions mirror the initial four causal estimands with,

\begin{enumerate}
    \item Incremental Dose response: $\theta_0^{\nabla:ATE}(a)=\mathbb{E}\{\nabla_a Y^{(a)}\}$.
     \item Incremental Dose response with distribution shift: $ \theta_0^{\nabla:DS}(a,\tilde{\text{\normalfont P}})=\mathbb{E}_{\tilde{\text{\normalfont P}}}\{\nabla_a Y^{(a)}\}$. 
    \item Incremental Conditional response: $ \theta_0^{\nabla:ATT}(a,a')=\mathbb{E}\{\nabla_{a'} Y^{(a')} \mid A=a\}$.
     \item Incremental Heterogeneous response: $\theta_0^{\nabla:CATE}(a,v)=\mathbb{E}\{\nabla_a Y^{(a)} \mid V=v\}$.
\end{enumerate}

Given the same problem assumptions, we can formulate this estimation as an integral via 

\begin{align}
    \theta^{\nabla:ATE}_0(a)  &= \int \nabla_a E(Y | A = a, X = x) d \text{P}(x)\\
    \theta^{\nabla:DS}_0(a) &= \int \nabla_a E(Y | A = a, X = x) d \tilde{\text{P}}(x)\\
    \theta^{\nabla:ATT}_0(a,a')  &= \int \nabla_{a'} E(Y | A = a', X = x) d \text{P}(x|a) \\
    \theta^{\nabla:CATE}_0(a,v) &= \int \nabla_a E(Y | A = a, V=v, X = x) d \text{P}(x|v).
\end{align}    
A similar process via Bochner integrability and the reproducing property yield,

\begin{align}
    \theta_0^{\nabla: ATE}(a) &=\bra{\gamma_0}\ket{ \nabla_a \phi(a)\otimes \mu_x}_{\mathcal{H}}\\
    \theta_0^{\nabla: DS}(a) &= \bra{\gamma_0}\ket{\nabla_a \phi(a)\otimes \nu_x}_{\mathcal{H}}\\
    \theta_0^{\nabla: ATT}(a,a') &= \bra{\gamma_0}\ket{\nabla_{a'} \phi(a') \otimes \mu_x(a)}_{\mathcal{H}}\\
    \theta_0^{\nabla: CATE}(a,v) &= \bra{\gamma_0}\ket{\nabla_a \phi(a)\otimes \phi(v) \otimes \mu_x(v)}_{\mathcal{H}}
\end{align}
where $\mu_x = \int \phi(x) d \text{pr}(x),\nu_x = \int \phi(x) \tilde{\text{pr}}(x) $ and $\mu_x(\cdot) = \int \phi(x) d \text{pr}(x|\cdot)$. Note that the boundedness and continuity of $\nabla_a \phi(a)$ is preserved.

Application of kernel ridge regression and approximate mean embeddings give us the four equations for incremental function estimation,

\begin{align}
\hat{\theta}^{ATE}(a) &= n^{-1} Y^T (K_{AA} \odot K_{XX} + n\lambda I)^{-1}( \nabla_a K_{Aa} \odot \sum_{x_i} K_{Xx_i}) \label{alg_app}\\
\hat{\theta}^{DS}(a) &= \tilde{n}^{-1} Y^T (K_{AA} \odot K_{XX} +  n\lambda I)^{-1}(\nabla_a  K_{Aa} \odot \sum_{\tilde{x_i}} K_{X\tilde{x_i}}) \\
\hat{\theta}_0^{ATT}(a,a') &= Y^T \qty(K_{AA} \odot K_{XX} + n\lambda I)^{-1} \qty[ \nabla_{a'} K_{Aa'} \odot \{K_{XX}(K_{AA} + n\lambda_{1}I)^{-1}K_{Aa}\}]\\
\hat{\theta}_0^{CATE}(a,v) &= Y^T \qty(K_{AA} \odot K_{VV} \odot K_{XX} + n\lambda I)^{-1} \nonumber \\& \qty[\nabla_a  K_{Aa} \odot K_{Vv} \odot \{K_{XX}(K_{VV} + n\lambda_{2}I)^{-1}K_{Vv}\}],
\end{align}
where $\nabla_a K_{Aa} = (\nabla_a k(A_1,a),\nabla_a k(A_2,a),\dots,\nabla_a k(A_n,a))$ gives a vector of derivatives of the kernel function. 

The same consistency and convergence bounds hold for these incremental functions except we replace $m_a$ with $m_a' \geq \sup_{a}\norm{\nabla_a \phi(a)}$. 

Consequently, the quantum algorithms are identical to \ref{alg 1} and \ref{alg 2} except we prepare the vector $\nabla_a K_{Aa}$ rather than $K_{Aa}$ and similarly for $a'$. We hence get the same consistency bounds as in appendix \ref{Section: Estimation Convergence with Quantum Algorithms} with additive errors from the quantum algorithms except we replace $m_a$ with $m_a'$.

\subsection{Counterfactual Distributions}\label{counterfactual distributions}

Our four counterfactual distributions are the same situations as presented in section \ref{Section: Classical Causal Inference Proofs} but instead of considering the \textit{mean} outcome, we now wish to estimate the full \textit{distribution} over the outcomes $y$. Counterfactual distributions can be encoded via a kernel mean embedding with feature map $\phi(y)$ for which all stated assumptions hold. The quantities of interest are identical to those in the main body of the paper but instead of estimating expectation values $E$, we wish to estimate the probability distributions,

\begin{enumerate}
    \item Counterfactual Distribution: $\theta_0^{D:ATE}(a)=\P\{ Y^{(a)}\}$.
     \item Counterfactual Distribution with distribution shift: $ \theta_0^{D:DS}(a,\tilde{\P})=\tilde{\P} \{ Y^{(a)}\}$. 
    \item Conditional Counterfactual Distribution: $ \theta_0^{D:ATT}(a,a')=\P \{ Y^{(a')} \mid A=a\}$.
     \item Heterogeneous Counterfactual Distribution: $\theta_0^{D:CATE}(a,v)=\P\{ Y^{(a)} \mid V=v\}$.
\end{enumerate}

We can similarly write these distributions as an integral given the standard assumptions from section \ref{Section: Classical Causal Inference Proofs},

\begin{align}
    \{\theta^{D:ATE}_0(a)\}(y)  &= \int \P(y | A = a, X = x) d \text{P}(x)\\
    \{\theta^{D:DS}_0(a)\}(y) &= \int \P(y | A = a, X = x) d \tilde{\text{P}}(x)\\
    \{\theta^{D:ATT}_0(a,a')\}(y)  &= \int \P(y | A = a', X = x) d \text{P}(x|a) \\
    \{\theta^{D:CATE}_0(a,v)\}(y) &= \int \P(y | A = a, V=v, X = x) d \text{P}(x|v).
\end{align}    

Importantly, to estimate these quantities Ref. \cite{grettonkernelcausal} show that we must first estimate the embeddings of counterfactual distributions given that $\phi(Y)$ satisfies assumption \ref{regularity},

\begin{align}
    \check{\theta}^{D:ATE}_0(a)  &= \int \mathbb{E} \{\phi(Y) | A = a, X = x \} d \text{P}(x)\\
    \check{\theta}^{D:DS}_0(a) &= \int \mathbb{E} \{\phi(Y) | A = a, X = x\} d \tilde{\text{P}}(x)\\
    \check{\theta}^{D:ATT}_0(a,a')  &= \int \mathbb{E}\{\phi(Y) | A = a', X = x\} d \text{P}(x|a) \\
    \check{\theta}^{D:CATE}_0(a,v) &= \int \mathbb{E}\{\phi(Y) | A = a, V=v, X = x\} d \text{P}(x|v).
\end{align} 

Next, given that we have defined the RKHS for $Y$, we can define the conditional expectation operator $E_3: \mathcal{H}_\mathcal{Y} \to \mathcal{H}_\mathcal{A} \otimes \mathcal{H}_\mathcal{X}, \phi(\cdot) \mapsto \mathbb{E}\{\phi(y) | A = \cdot, X = \cdot \}$. Then we note the adjoint of $E_3$, denoted by $E_3^*$, is an operator for $\mathcal{H}_\mathcal{A} \to \mathcal{H}_\mathcal{Y}$. By construction, $E_3^*\{\phi(a)\otimes \phi(x)\} = \mathbb{E}(\phi(Y)|A=a,X=x)$. And similarly for CATE, where $X$ is replaced with $(V,X)$. Given regularity assumptions on this operator, immediately satisfied in the quantum case as operators are in $\mathcal{L}_2$, we can decouple integration and expectation with,

\begin{align}
    \check{\theta}_0^{D:ATE}(a) &= E_3^* \left\{ \phi(a) \otimes \mu_x \right\}, \quad \text{where} \ \mu_x = \int \phi(x) dpr(x).\\
    \check{\theta}_0^{D:DS}(a, \tilde{pr}) &= E_3^* \left\{ \phi(a) \otimes \nu_x \right\}, \quad \text{where} \ \nu_x = \int \phi(x) d\tilde{pr}(x).\\
    \check{\theta}_0^{D:ATT}(a, a') &= E_3^* \left\{ \phi(a') \otimes \mu_x(a) \right\}, \quad \text{where} \ \mu_x(a) = \int \phi(x) dpr(x | a).\\
    \check{\theta}_0^{D:CATE}(a, v) &= E_3^* \left\{ \phi(a) \otimes \phi(v) \otimes \mu_x(v) \right\}, \quad \text{where} \ \mu_x(v) = \int \phi(x) dpr(x | v).
\end{align}

The estimation of these conditional embeddings is done via kernel ridge regressions of the kernel mean embeddings $\hat{\mu},\hat{\nu}$ and the adjoint conditional expectation operator $\hat{E_3^*}$. This leads to the closed form solutions, with proof given in Supplement E of \cite{grettonkernelcausal}, 

\begin{align}
\{\hat{\theta}_0^{D:ATE}(a)\}(y) &= n^{-1} \sum_{i=1}^n K_{yY}(K_{AA} \odot K_{XX} + n\lambda_3I)^{-1}(K_{Aa} \odot K_{Xx_i}); \\
\{\hat{\theta}_0^{D:DS}(a)\}(y) &= \tilde{n}^{-1} \sum_{i=1}^{\tilde{n}} K_{yY}(K_{AA} \odot K_{XX} + n\lambda_3I)^{-1}(K_{Aa} \odot K_{X\tilde{x}_i}); \\
\{\hat{\theta}_0^{D:ATT}(a,a')\}(y) &= K_{yY}(K_{AA} \odot K_{XX} + n\lambda_3I)^{-1}[K_{Aa'} \odot \{K_{XX}(K_{AA} + n\lambda_1I)^{-1}K_{Aa}\}]; \\
\{\hat{\theta}_0^{D:CATE}(a,v)\}(y) &= K_{yY}(K_{AA} \odot K_{VV} \odot K_{XX} + n\lambda_3I)^{-1}\nonumber [K_{Aa} \odot K_{Vv} \\& 
 \odot \{K_{XX}(K_{VV} + n\lambda_2I)^{-1}K_{Vv}\}];
\end{align}

This again implies identical quantum algorithms to the main text except with the replacements $\lambda_3$ for $\lambda$ and $K_{yY}$ for $Y$. Hence the consistency bounds are equivalent to those in Appendix \ref{Section: Estimation Convergence with Quantum Algorithms} with additive errors but we replace $\norm{Y}$ with the maximum kernel entry $m_y$ and the constants in the classical consistency change from $(c,b)$ to $(c_3,b_3)$ to satisfy smoothness in $\mathcal{Y}$. 

These estimations of embeddings can be used to generate samples $(\tilde{Y}_j)$ which weakly converge to any of the four desired distributions $\theta_0^{D:}(a)$. This is also true for the quantum algorithm but would require an exponentially greater number of samples due to the sampling error. This is due to the difficulty of sampling from quantum algorithms as demonstrated in section \ref{Section: Convergence Rate with Quantum Measurement Error}.

\subsection{Graphical Models}\label{graphical models}

The third addition to the primary results of counterfactual response curves is that of graphical models via Pearl's \textit{do} calculus method \cite{JudeaPearl1995}. This model generalises the prior cases to conditions when there are unobserved confounders $U$ which may interact with $X,D,Y$. In particular, we wish to estimate,

\begin{enumerate}
    \item Counterfactual mean outcome given intervention: $\theta_0^{do}(a) = \mathbb{E}\{Y|do(D=d)\}$
    \item Counterfactual distribution given intervention: $\theta_0^{D:do}(a) = \P\{Y | do(D=d)\}$
    \item Counterfactual distribution embedding given intervention: $\check{\theta}_0^{D:do}(a) = \mathbb{E}\{\phi(Y) | do(D=d)\}$.
\end{enumerate}

Ref. \cite{grettonkernelcausal} present an argument showing $\hat{\theta}^{ATE}$ is a uniformly consistent estimator of such quantities if the back door criterion hold. Hence, they present a novel formulation for the front door criterion, which assumes,

\begin{enumerate}
    \item $X$ intercepts all directed paths from $D$ to $Y$.
    \item There is no unblocked back door path from $D$ to $X$
    \item All back door paths from $X$ to $Y$ are blocked by $D$
    \item $pr(D,X) > 0$ almost surely.  
\end{enumerate}

Given said assumptions one may express the quantities of interest as integrals, 

\begin{align}
    \theta_0^{do}(a) &= \int \mathbb{E}(Y|A=a',X=x) d\P(a') d\P(x|a)\\
    \check{\theta}_0^{D:do}(a) &= \int \mathbb{E}\{\phi(Y)|A=a',X=x\} d\P(a') d\P(x|a).
\end{align}

Given a learnable regression $\mathbb{E}(Y,A=a',X=x)$ and a regularised conditional expectation operator $E_3$, the estimations of such quantities can be found by,

\begin{align}
    \hat{\theta}^{FD}(a) &= \langle \hat{\gamma},\hat{\mu}_a \otimes \hat{\mu}_x(a) \rangle\\
    \hat{\theta}^{D:FD}(a) &= \hat{E}_3^* \{ \hat{\mu}_a \otimes \hat{\mu}_x(a) \}\\
\end{align}
where $\mu_a = \int \phi(a) d\P(a)$, and the estimator of each quantity is found via kernel ridge regression. This yields the closed form solutions,

\begin{align}
\hat{\theta}^{FD}(a) &= n^{-1} \sum_{i=1}^n Y^\top (K_{AA} \odot K_{XX} + n\lambda I)^{-1} [K_{Aa_i} \odot \{K_{XX}(K_{AA} + n\lambda_1 I)^{-1} K_{Aa}\}] \\
\{\hat{\theta}^{D:FD}(a)\}(y) &= n^{-1} \sum_{i=1}^n K_{yY} (K_{AA} \odot K_{XX} + n\lambda_3 I)^{-1} [K_{Aa_i} \odot \{K_{XX}(K_{AA} + n\lambda_1 I)^{-1} K_{Aa}\}].
\end{align}

The quantum form of these algorithms mirror algorithm \ref{alg 2} except for the substitutions $K_{Aa} \to n^{-1}\sum_{i=1}^n K_{Aa_i}$ and $Y^T \to K_{yY}$. Again, the same tradeoff between quantum speedup and slower convergence remains due to the additive error of inner product evaluation as in Appendix \ref{Section: Estimation Convergence with Quantum Algorithms}. Samples can again be generated from the estimator of the distribution embeddings which weakly approach the true distribution ($\theta_0^{D:do}(a)$).

\end{document}